\documentclass[12pt,notitlepage]{article}%
\usepackage{amsthm}
\usepackage{setspace}
\usepackage{eurosym}
\usepackage{subfigure}
\usepackage{xcolor}
\usepackage{amssymb}
\usepackage{graphicx}
\usepackage[colorlinks,linkcolor=links,citecolor=cites,urlcolor=MyDarkBlue]%
{hyperref}
\usepackage{amsfonts}
\usepackage{amsmath}
\usepackage{amstext}
\usepackage{appendix}
\usepackage{appendix}
\usepackage{mathpazo}
\usepackage{tikz}
\usepackage{verbatim}

\usetikzlibrary{trees}
\usetikzlibrary{matrix}
\usetikzlibrary{decorations.pathreplacing}
\usetikzlibrary{calc}
\usepackage{natbib}%
\providecommand{\U}[1]{\protect\rule{.1in}{.1in}}
\newtheorem{theorem}{Theorem}

\newtheorem{definition}{Definition}
\newtheorem{example}{Example}

\newtheorem{lemma}{Lemma}

\numberwithin{equation}{section}
\definecolor{MyDarkBlue}{rgb}{0,0.08,0.45}
\definecolor{cites}{HTML}{324b13}
\definecolor{links}{HTML}{1a663b}
\definecolor{MyLightMagenta}{cmyk}{0.1,0.8,0,0.1}
\hypersetup{
colorlinks,citecolor=blue,filecolor=black,linkcolor=blue,urlcolor=blue
}
\usepackage[top=1.25in,bottom=1.25in,left=1.25in,right=1.25in]{geometry}
\usepackage{setspace}
\usepackage{footmisc}
\usepackage{lipsum}
\begin{document}
\title{Stable matching: an integer programming approach\thanks{I thank the coeditor, Federico Echenique, and two anonymous referees for very helpful comments and suggestions that greatly improved this paper. I thank Zhenhua Jiao, Yan Ju, Wei Ma, and Jun Zhang for many helpful discussions. I am grateful to Ning Sun, Qianfeng Tang, and Guoqiang Tian for their continuous support and encouragement. All errors are mine.}}
\author{Chao Huang\thanks{Institute for Social and Economic Research, Nanjing Audit University. Email: huangchao916@163.com}}
\date{}
\maketitle

\begin{abstract}
This paper develops an integer programming approach to two-sided many-to-one matching by investigating stable integral matchings of a fictitious market where each worker is divisible. We show that stable matchings exist in a discrete matching market when firms' preference profile satisfies a total unimodularity condition that
is compatible with various forms of complementarities. We provide a class of firms' preference profiles that satisfy this condition.
\end{abstract}



\section{Introduction}
Studies of two-sided matching originated with the seminal work of \cite{GS62} on the markets of marriage and college admission. In the past decades, the theory of two-sided matching has provided practical solutions to real-life matching problems such as hospital-doctor matching, college admission, and school choice.\footnote{See \cite{RP99}, \cite{R02}, \cite{AS03}, \cite{APR09}, \cite{EY15}, and \cite{KK15}, among others.} \cite{KC82} and \cite{R84} have recognized different versions of substitutability conditions that are sufficient for the existence of a stable matching in different settings. The substitutability condition for a discrete matching market\footnote{We use the term ``discrete matching market" to distinguish it from the matching market with continuous monetary transfers and quasi-linear utilities; see Section \ref{lite}.} requires that any worker chosen by a firm from a set of available workers should still be chosen when the available set shrinks.\footnote{If a worker is chosen by a firm from a set of available workers but not chosen when the available set shrinks, we would deem that some of the worker's complements become unavailable as the set shrinks.} Most real-life matching practices rely on the substitutability condition. Complementarities in firms' preferences have been regarded as the primary source of difficulties for market design. For instance, consider the following market borrowed from \cite{CKK19} (henceforth, CKK)\footnote{See Section 2 of CKK. Literature on the quasi-linear market also uses examples of a similar structure to illustrate the nonexistence of an equilibrium caused by complementarities; see e.g., \cite{AWW13} and \cite{BK19}.} where there are two firms $f_1,f_2$ and two workers $w_1,w_2$. The agents have the following preferences.

\begin{equation}\label{exam_in1}
\begin{aligned}
&f_1: \{w_1,w_2\}\succ\emptyset \qquad\qquad\qquad\qquad &w_1: &\quad f_1\succ f_2\\
&f_2: \{w_1\}\succ\{w_2\}\succ\emptyset \qquad\qquad\qquad\qquad &w_2: &\quad f_2\succ f_1
\end{aligned}
\end{equation}
Firm $f_1$'s preference violates the substitutability condition since $w_1$ and $w_2$ are complements for $f_1$.\footnote{This is because $w_1$ would be hired by $f_1$ from $\{w_1,w_2\}$ but not hired by $f_1$ from $\{w_1\}$. Similarly, $w_2$ would be hired by $f_1$ from $\{w_1,w_2\}$ but not hired by $f_1$ from $\{w_2\}$.} No stable matching exists in this market.\footnote{In this market, $f_1$ should hire both workers or neither in any stable matching. In the former case, $f_2$ would form a blocking coalition with $w_2$, who prefers $f_2$ to $f_1$. In the latter case, $f_2$ would be matched with $w_1$, leaving $w_2$ unmatched. Then, $f_1$ would form a blocking coalition with both $w_1$ and $w_2$.}

The lack of stable matching has been attributed to complementarities in firms' preferences. However, consider the following market where $f_2$'s preference is changed.

\begin{equation}\label{exam_in2}
\begin{aligned}
&f_1: \{w_1,w_2\}\succ\emptyset \quad\qquad &w_1: &\quad f_1\succ f_2\\
&f_2: \{w_1,w_2\}\succ\{w_1\}\succ\{w_2\}\succ\emptyset \quad\qquad &w_2: &\quad f_2\succ f_1
\end{aligned}
\end{equation}
There exists a stable matching in this market where $f_2$ hires both $w_1$ and $w_2$. The complementarity in $f_1$'s preference does not distort stability in this market. This paper shows that the firms' preference profile in the latter example satisfies a total unimodularity condition that guarantees a stable matching for all possible preferences of workers, whereas this condition is violated in the former example. We find that the notion of demand type proposed by \cite{BK19} (henceforth, BK) is useful for analyzing firms' preferences in a discrete matching market.\footnote{BK studied the exchange economy where agents have quasi-linear utilities over indivisible goods, which we call the quasi-linear market. See Section \ref{lite}.} When the available set of workers expands, $f_1$ would hire both workers when they are available, thus $f_1$ has a demand type of $\{(1,1)\}$. When the available set of workers for $f_2$ expands from $\{w_2\}$ to $\{w_1,w_2\}$, $f_2$ would drop $w_2$ and hire $w_1$ in the former example, thus the demand type of $f_2$ contains $(1,-1)$, whereas the demand type of $f_2$ does not contain $(1,-1)$ in the latter example. The demand type of $f_2$ is $\{(1,0),(0,1)(1,-1)\}$ in the former example, and $\{(1,1),(1,0),(0,1)\}$ in the latter example.\footnote{In both examples, when the available set of workers for $f_2$ expands from $\emptyset$ to $\{w_1\}$ or $\{w_2\}$, $f_2$ would hire $w_1$ or $w_2$. Thus $(1,0)$ and $(0,1)$ are in $f_2$'s demand type in both examples. In the latter example, the demand type of $f_2$ also contains $(1,1)$, since when the available set of workers expands from $\emptyset$ to $\{w_1,w_2\}$, $f_2$ would hire both workers.} A matrix is totally unimodular if every square submatrix has determinant 0 or $\pm1$. A set of vectors is totally unimodular if the matrix that has these vectors as columns is totally unimodular. The firms' demand type is $\{(1,1),(1,0),(0,1),(1,-1)\}$ in the former example, which is not totally unimodular (because the determinant of $(1,1)$ and $(1,-1)$ is $-2$). The firms' demand type is $\{(1,1),(1,0),(0,1)\}$ in the latter example, which is totally unimodular (because any matrix formed by two of these vectors has determinant 0 or $\pm1$). In this paper, we show that this total unimodularity of firms' preferences guarantees the existence of a stable matching for all possible preferences of workers.

Previous studies on discrete matching markets are mostly based on the Gale-Shapley mechanism or Taski's fixed-point theorem.\footnote{The Gale-Shapley mechanism with its variants is a standard tool for two-sided matching. See Section \ref{lite} for literature on the fixed-point method.} This paper develops a new integer programming approach. We prove our results by studying stable integral matchings of a continuum market induced from the original market. We assume that each worker is divisible and construct a particular continuous preference on worker shares for each firm. This continuum market is an instance of CKK's continuum market, and CKK's existence theorem guarantees the existence of a stable matching in our continuum market. Each stable integral matching of this market corresponds to a stable matching of the original market. Finally, we can apply tools from integer programming to prove the existence of a stable integral matching in the continuum market when firms' demand type is totally unimodular. In a discrete matching market, total unimodularity is independent of substitutability and compatible with various complementarities.\footnote{In a quasi-linear market, substitutability implies total unimodularity (see Section 3.2 of BK). However, the two conditions are independent of each other in a discrete matching market, see Section \ref{Sec_theorem}.} Therefore, our result indicates that stable matching is compatible with various forms of complementarities in firms' preferences. This observation runs contrary to the common belief that complementarities distort stability.

Our existence theorem applies to a problem of matching firms with specialists, which has a practical economic meaning. We depict the structure of firms' acceptable sets of workers by a technology tree, where each vertex of the tree represents a technology that requires a set of workers to implement. Each edge of the tree is an upgrade from one technology to another that requires more workers. We say that a worker is a specialist if she engages in only one upgrade. We show that firms' demand type is totally unimodular when firms have preferences over the technologies in a technology tree where each worker is a specialist. See Section \ref{Sec_App}.

\subsection{Related literature}\label{lite}

There are two parallel lines in the literature of matching theory. The discrete matching market we study assumes that there are no monetary transfers (school choice and college admission) or that workers' wages are exogenously given (firm-worker matching and hospital-doctor matching).
The other line assumes that there are continuous monetary transfers between firms and workers where firms have quasi-linear utilities and workers' wages are determined endogenously. This market can be viewed as an exchange economy where agents have quasi-linear utilities over indivisible goods.\footnote{See Section 6.4 of BK.} We call both the matching market and the exchange econmy of this line the quasi-linear market, where the solution concept may be stable matching or equilibrium.

\cite{KC82} found that stable matchings exist in a quasi-linear market when each firm has a gross-substitute valuation over workers. \cite{R84,R85} found that stable matchings exist in a discrete many-to-many market when all agents have substitutable preferences. In quasi-linear markets and discrete matching markets, \cite{SY06,SY09} and \cite{O08} respectively generalized the restriction of substitutability to allow complementarities between two group of agents where substitutability is satisfied within each group.\footnote{\cite{O08} studied the problem of supply chain networks, which subsumes the two-sided matching as a special case; see also \cite{HKNOW13}.} \cite{DKM01} showed that a unimodularity condition is  sufficient for the existence of an equilibrium in a quasi-linear market.\footnote{The proof was provided in \cite{DK04} as a special case of Theorem 3 therein.} BK proved this result independently and enhanced it into the unimodularity theorem, which states that an equilibrium exists for all profiles of concave valuations of a demand type if and only if the demand type is unimodular.

This paper is closely related to CKK, with BK's notion of demand type playing a critical role. In a discrete matching market, CKK found that stable matchings always exist with a continuum of workers when firms have continuous preferences, which can exhibit various complementarities. Their approach enables market designers to pursue an approximately stable matching, which is guaranteed to exist in a real-life market when the market is large enough. CKK proved their existence theorem for the continuum market using the Brouwer (or Kakutani) fixed-point theorem.\footnote{CKK applied the Kakutani-Fan-Glicksberg theorem to prove their existence theorem in a general setting, allowing for indifferences in firms' preferences and infinite worker types. Our continuum market is the basic case with strict preferences and finite worker types in which stable matchings correspond to Brouwer fixed points.} Thus, the existence of a stable matching proved in this paper is essentially the existence of an integral Brouwer fixed point, which is quite different from previous methods in matching theory.

BK's demand type considers how demand changes as prices change, while our demand type considers how demand changes as available workers change. BK proved their results using tools from tropical geometry. \cite{TY19} showed that the unimodularity theorem for the quasi-linear market could also be proved via integer programming. Although we adopt the notion of demand type from BK and use a similar tool as \cite{TY19}, our method is quite different from theirs. The unimodularity theorem for the quasi-linear market is built on the following result: Equilibrium always exists in a quasi-linear market if and only if the aggregate valuation is concave.\footnote{See e.g., Lemma 2 of \cite{TY19}.} In a discrete matching market, there is no counterpart to the notion of the aggregate valuation. Therefore, the methods for quasi-linear markets do not apply to discrete matching markets. Our method is to study the existence of a stable integral matching in a market with ``divisible" workers, where stable matchings always exist according to CKK.

Our method also relates to the linear programming (henceforth, LP) method for selecting stable matchings of specific properties. \cite{V89} showed the polytope defined by the stable matchings in the marriage market and solved the optimal marriage problem as a linear program. This method was then developed by \cite{R92}, \cite{RRV93}, and \cite{BB00}. \cite{TS98} and \cite{STQ06} used this method to find a median stable matching, which is a compromise between agents of the two sides among stable matchings. We discuss the relation between the LP method and our method in Section \ref{Sec_LP}. \cite{BMM14} and \cite{ABM16} used integer programming methods to study complexities of matching with couples and special college admission problems, respectively.

Stable matchings in a discrete matching market have been characterized as Taski's fixed points; see \cite{A00}, \cite{F03}, \cite{EO04,EO06}, and \cite{HM05}, among others. A more straightforward method for our purpose would be studying the system defined by the fixed-point characterization under the total unimodularity condition. The fixed-point method is appealing for our problem because the fixed point characterization provides a necessary and sufficient condition that can be used for studying further generalizations or variations of the total unimodularity condition.\footnote{For instance, we can not tell whether a stable matching exists under some generalized unimodularity condition if this condition does not guarantee an integral solution to the constructed linear equations in Section \ref{Sec_method}. By contrast, we can confirm that stable matching does not exist when there is no fixed point for the related operator in the fixed point method.} Although it is so far unknown how to obtain our result using the fixed-point method, proving our result using the fixed-point method would be useful for future research.

Complementarities were also considered in the problem of matching with couples (e.g., \citealp{KK05} and \citealp{NV18}) and matching with peer effects (see \citealp{EY07} and \citealp{P12}). CKK is a part of the literature that treated the problem of complementarity by pursuing approximate stable outcomes; see also \cite{AWW13}, \cite{KPR13}, and \cite{AH18}, among others. The setting of discrete matching has been generalized to allow discrete contract terms between firms and workers (see \citealp{R84} and \citealp{HM05}), where contracts may specify wages, insurances, retirement plans, etc. Under the generalized setting, \cite{HK10} proposed weakened substitutability conditions that are not implied by the Kelso-Crawford gross substitutes condition (see \citealp{E12}); see also \cite{HK19}. We can also extend our model to the framework of matching with contracts by studying stable integral matchings of the market in Section S.9 of CKK.\footnote{CKK extended their results to the framework of matching with contract terms in their Section S.9.} The author's recent works, \cite{H21a} and \cite{H21b}, provide sufficient conditions for stable matching where firms' preferences may violate both the substitutability condition and the total unimodularity condition.

The remainder of this paper is organized as follows. Section Section \ref{Sec_M} presents the model of a discrete matching market and the existence theorem. Section \ref{Sec_method} elaborates on our method for the proof of the existence theorem. Section \ref{Sec_App} provides an application of matching firms with specialists. Proofs are relegated to the Appendix.

\section{Model\label{Sec_M}}

\subsection{Preliminaries}

There is a set $F=\{f_1,\ldots,f_m\}$ of $m$ firms, and a set $W=\{w_1,\ldots,w_n\}$ of $n$ workers. Let ${\o}$ be the null firm, representing not being matched with any firm. Each worker $w\in W$ has a strict, transitive, and complete preference $\succ_w$ over $\widetilde{F}:=F\cup\{{\o}\}$. For any $f, f'\in \widetilde{F}$, we write $f\succ_w f'$ when $w$ prefers $f$ to $f'$ according to $\succ_w$. We write $f\succeq_w f'$ if either $f\succ_w f'$ or $f=f'$. Let $\succ_W$ denote the preference profile of all workers. Each firm $f\in F$ has a strict, transitive and complete preference $\succ_f$ over $2^W$. For any $S,S'\subseteq W$, we write $S\succ_f S'$ when $f$ prefers $S$ to $S'$ according to $\succ_f$. We write $S\succeq_f S'$ if either $S\succ_f S'$ or $S=S'$. Let $\succ_F$ be the preference profile of all firms. A matching market can be summarized as a tuple $\Gamma=(W,F,\succ_W,\succ_F)$.

Let $Ch_f$ be the choice function of $f$ such that for any $S\subseteq W$, $Ch_f(S)\subseteq S$ and $Ch_f(S)\succeq_f S'$ for any $S'\subseteq S$. By convention, let $Ch_{{\o}}(S)=S$ for all $S\subseteq W$. For any $f\in F$, any $w\in W$, and any $S\subseteq W$, we say that $f$ is acceptable to $w$ if $f\succ_w {\o}$; we say that $S$ is acceptable to $f$ if $S\succ_f \emptyset$.\footnote{Note that ${\o}$ is the null firm whereas $\emptyset$ is the empty set of workers.}

\begin{definition}
\normalfont
A \textbf{matching} $\mu$ is a function from the set $\widetilde{F}\cup W$ into $\widetilde{F}\cup 2^W$ such that for all $f\in \widetilde{F}$ and $w\in W$,
\begin{description}
\item[(\romannumeral1)] $\mu(w)\in \widetilde{F}$;

\item[(\romannumeral2)] $\mu(f)\in 2^W$;

\item[(\romannumeral3)] $\mu(w)=f$ if and only if $w\in\mu(f)$.
\end{description}
\end{definition}

We say that a matching $\mu$ is \textbf{individually rational} if $\mu(w)\succeq_w {\o}$ for all $w\in W$ and $\mu(f)=Ch_f(\mu(f))$ for all $f\in F$. We say that a firm $f$ and a subset of workers $S\subseteq W$ form a \textbf{blocking coalition} that blocks $\mu$ if $f\succeq_w\mu(w)$ for all $w\in S$, and $S\succ_f\mu(f)$. In words, individual rationality requires that each matched worker prefer her current employer to being unmatched and that no firm wish to unilaterally drop any of its employees. When $f$ and $S$ block $\mu$, $S$ may contain workers that are matched with $f$ in $\mu$. Thus, we require each worker $w\in S$ to weakly prefer $f$ to $\mu(w)$. $f$ should strictly prefer $S$ to $\mu(f)$ since we require $S\neq\mu(f)$.

\begin{definition}\label{stability}
\normalfont
A matching $\mu$ is \textbf{stable} if it is individually rational and there is no blocking coalition that blocks $\mu$.\footnote{The no blocking coalition condition implies the individual rationalities of firms: $\mu(f)=Ch_f(\mu(f))$ for all $f\in F$. This is also the case for Definition \ref{stable} (see footnote 28 of CKK). The exposition of these two definitions follows the convention in the literature.}
\end{definition}

Stable matching is guaranteed to exist in a matching market when the preference of each firm satisfies the following substitutability condition; see Chapter 6 of \cite{RS90}.\footnote{The stability of Definition \ref{stability} is called the core defined by weak domination in \cite{RS90}.}

\begin{definition}
\normalfont
Firm $f$ has a \textbf{substitutable} preference if for any $S\subseteq W$ and any$\{w,w'\}\subseteq S$, $w\in Ch_f(S)$ implies $w\in Ch_f(S\setminus \{w'\})$.
\end{definition}

If, on the contrary, there exist $S\subseteq W$ and $\{w,w'\}\subseteq S$ such that $w\in Ch_f(S)$ but $w\notin Ch_f(S\setminus \{w'\})$, then we would consider $w$ and $w'$ to be complements in $f$'s preference because $w$ becomes undesired when $w'$ is not available.

\subsection{Demand type}

Now we introduce BK's concept of demand type to discrete matching markets. For any subset $S\subseteq W$ of workers, we let $ind(S)\in\{0,1\}^W$ denote the indicator vector of $S$. For any $S, S'\subseteq W$ such that $ind(S)=\mathbf{x}$ and $ind(S')=\mathbf{x}'$, we abuse the notation to write $\mathbf{x}\succeq_f \mathbf{x}'$ when $S\succeq_f S'$, $\mathbf{x}\succ_f \mathbf{x}'$ when $S\succ_f S'$, and $Ch_f(\mathbf{x})=\mathbf{x}'$ when $Ch_f(S)=S'$.

\begin{definition}\label{df_dt}
\normalfont
For each $f\in F$, let $\mathcal{D}_f=\{\mathbf{d}\in \{0,1\}^W\mid  \mathbf{d}\neq \mathbf{0}$ and $\mathbf{d}=ind(Ch_f(S))-ind(Ch_f(S'))$ for some $S,S'$ such that $S'\subset S\subseteq W\}$ be $f$'s \textbf{demand type}.
The demand type for the firms' preference profile is $\mathcal{D}=\cup_{f\in F}\mathcal{D}_f$.
\end{definition}

A matrix is \textbf{totally unimodular} if every square submatrix has determinant 0 or $\pm1$. In particular, each entry in a totally unimodular matrix is 0 or $\pm1$. A set of vectors is called totally unimodular if the matrix that has these vectors as columns is totally unimodular. Total unimodularity of firms' demand type can be tested in polynomial time; see, for example, \cite{WT13}. In a practical problem, suppose a firm $f$ reports a preference that contains $N$ acceptable sets. If $S$ and $S'$ are two sets such that $Ch_f(S')=S'$ and $S\succ_f S'$, we should check whether $S=Ch_f(S\cup S')$ holds to determine whether $ind(S)-ind(S')$ belongs to $f$'s demand type. Hence, it takes $\mathcal{O}(N^3)$ time to obtain $f$'s demand type, which contains at most $N(N+1)/2$ vectors.

\subsection{Existence theorem}\label{Sec_theorem}

We find that stable matchings exist in a discrete matching market for all possible preferences of workers, provided the firms' demand type is totally unimodular.\footnote{A previous version of this paper (\citealp{H21c}) defined the null firm's demand type to be the set of unit vectors and defined the firms' demand type to include the null firm's demand type. The unimodularity of firms' demand type defined in the previous version is equivalent to the current total unimodularity condition according to the following fact: For any matrix $H$, $[H,I]$ is unimodular if and only if $H$ is totally unimodular.}

\begin{theorem}\label{thm_main}
\normalfont
There exists a stable matching if firms' demand type $\mathcal{D}$ is totally unimodular.
\end{theorem}

We explain our method for the proof of this theorem in Section \ref{Sec_method}. Totally unimodular demand types are prevalent and compatible with various complementarities. We provide a class of firms' preference profiles in Section \ref{Sec_App} that exhibits both total unimodularity and complementarities.

\cite{DKM01} and BK showed that an equilibrium exits in a quasi-linear market if agents have a unimodular demand type. Unimodularity is weaker than total unimodularity. A set of vectors in $\mathbb{Z}^n$ is \textbf{unimodular} if every linearly independent subset can be extended to a basis for $\mathbb{R}^n$, of integer vectors, with determinant $\pm1$.\footnote{$k<n$ linearly independent vectors $\mathbf{a}^1,\ldots,\mathbf{a}^k$ can be extended to a basis for $\mathbb{R}^n$ if there exist vectors $\mathbf{a}^{k+1},\ldots,\mathbf{a}^n$ such that $\mathbf{a}^1,\ldots,\mathbf{a}^n$ is a basis for $\mathbb{R}^n$. By the ``determinant" of $n$ vectors in $\mathbb{Z}^n$ we mean the determinant of the $n\times n$ matrix that has them as its columns.} A matrix is unimodular if the set of its columns is unimodular. Notice that a matrix of full row rank is unimodular if any square submatrix formed by its columns has determinant $0$ or $\pm1$. The following example shows that a stable matching does not necessarily exist in a discrete matching market when $\mathcal{D}$ is unimodular.\footnote{Example \ref{exam_uni} does not conflict with the existence theorem stated in the previous version (\citealp{H21c}) because the definition of firms' demand type in the previous version differs from Definition \ref{df_dt}.}

\begin{example}\label{exam_uni}
\normalfont
There are three firms $f_1,f_2,f_3$, and three workers $w_1,w_2,w_3$. The agents have the following preferences.
\begin{equation}\label{exam_unimodular}
\begin{aligned}
&f_1: \{w_1,w_2,w_3\}\succ\emptyset \qquad\qquad\qquad\qquad &w_1: &\quad f_1\succ f_2\succ {\o}\\
&f_2: \{w_1\}\succ\{w_2\}\succ\emptyset \qquad\qquad\qquad\qquad &w_2: &\quad f_2\succ f_1\succ f_3\succ {\o}\\
&f_3: \{w_2,w_3\}\succ\emptyset \qquad\qquad\qquad\qquad &w_3: &\quad f_1\succ f_3\succ {\o}
\end{aligned}
\end{equation}
We have $\mathcal{D}_{f_1}=\{(1,1,1)\}$, $\mathcal{D}_{f_2}=\{(1,0,0),(0,1,0),(1,-1,0)\}$, and $\mathcal{D}_{f_3}=\{(0,1,1)\}$. The firms' demand type is $\mathcal{D}=\{(1,1,1),(1,0,0),(0,1,0),(1,-1,0),(0,1,1)\}$, which is unimodular but not totally unimodular.\footnote{$\mathcal{D}$ is unimodular since the matrix formed by vectors from $\mathcal{D}$ has full row rank and any submatrix formed by three vectors from $\mathcal{D}$ has determinant $0$ or $\pm1$. $\mathcal{D}$ is not totally unimodular because of the submatrix formed by $(1,1)$ and $(1,-1)$.} This market does not admit a stable matching: Suppose $f_3$ is matched with $\{w_2,w_3\}$, then $w_1$ is either employed by $f_2$ or unmatched. In both cases, $f_1$ and $\{w_1,w_2,w_3\}$ form a blocking coalition. Suppose $f_3$ is matched with $\emptyset$, we do not have a stable matching for a similar reason as (\ref{exam_in1}).
\end{example}

In a quasi-linear market, both conditions of gross substitutes (\citealp{KC82}) and gross substitutes and complements (\citealp{SY06}) imply total unimodularity.\footnote{See Section 3.2 of BK.} However, total unimodularity is independent of substitutability in a discrete matching market. In a quasi-linear market, the demand type for a gross-substitute valuation only involves vectors that each has at most one coordinate of $+1$ and at most one coordinate of $-1$.\footnote{See Definition 3.5 and Proposition 3.6 of BK.} But this is not the case for the demand type of a firm's substitutable preference in a discrete matching market. Next is an example of a substitutable preference profile with demand type that fails total unimodularity.

\begin{example}\label{last}
\normalfont
There are two firms $f_1,f_2$, and two workers $w_1,w_2$. The firms have the following preferences.
\begin{align*}
&f_1: \{w_1,w_2\}\succ\{w_1\}\succ\{w_2\}\succ\emptyset \qquad \qquad \qquad f_2: \{w_1\}\succ\{w_2\}\succ\emptyset
\end{align*}
Both firms have substitutable preferences.
However, we have $\mathcal{D}_{f_1}=\{(1,1),(1,0),(0,1)\}$ and $\mathcal{D}_{f_2}=\{(1,0),(0,1),(1,-1)\}$. The firms' demand type is $\mathcal{D}=\{(1,1),(1,0),(0,1),(1,-1)\}$, which is not totally unimodular.
\end{example}

The necessity part of BK's unimodularity theorem says that given a demand type that is not unimodular, there must be some profile of concave valuations of this demand type for which an equilibrium does not exist.\footnote{See Corollary 4.4 of BK.} It is unknown whether the following counterpart holds in our context: In a discrete matching market $\Gamma$ with firms' demand type $\mathcal{D}$ not being totally unimodular (such as the market in Example \ref{last}), there must be some market $\Gamma'$ with the same demand type for firms, for which stable matching does not exist (such as market (\ref{exam_in1}), the firms' demand type in (\ref{exam_in1}) is the same as that in Example \ref{last}).

\section{Method}\label{Sec_method}

We elaborate on our method for the proof of Theorem \ref{thm_main} in this section. We present the market with ``divisible" workers in Section \ref{Sec_cont}, the stability-preserving turnovers in this market in Section \ref{Sec_preserve}, and an illustrative example in Section \ref{Sec_illu}. We discuss the relation to the LP method for stable matching in Section \ref{Sec_LP}.

\subsection{A continuum market}\label{Sec_cont}

We construct an instance of the continuum market in CKK from a matching market $\Gamma$ by assuming that each worker in $W$ is divisible. In this setting, $W$ is called the set of worker types. There is a divisible mass of quantity 1 of each worker type $w\in W$. A vector $\mathbf{x}\in[0,1]^W$ is called a \textbf{subpopulation} where $x(w)$ is the quantity of the type-$w$ workers for each $w\in W$. We say that a subset of worker types $S\subseteq W$ (or $\mathbf{z}\in\{0,1\}^W$) is available at subpopulation $\mathbf{x}\in[0,1]^W$ if $w\in S$ implies $x(w)>0$ (resp. $z(w)=1$ implies $x(w)>0$). Each worker type $w\in W$ has a strict preference $\succ_w$ over $\widetilde{F}$, where $\succ_w$ is the preference of worker $w$ in the original market $\Gamma$. Each firm $f\in \widetilde{F}$ has a choice function $\widehat{Ch}_f: [0,1]^W\rightarrow [0,1]^W$, where $\widehat{Ch}_f(\mathbf{x})\leq \mathbf{x}$. By convention, let $\widehat{Ch}_{\o}(\mathbf{x})=\mathbf{x}$ for all $\mathbf{x}\in [0,1]^W$. We now construct a choice function $\widehat{Ch}_f$ for each firm $f\in F$ based on its preference $\succ_f$ in $\Gamma$.

Let $\mathbf{u}^1\succ_f\mathbf{u}^2\succ_f\cdots\succ_f\mathbf{u}^L\succ_f\mathbf{0}$ be the preference order of firm $f$ over its acceptable sets of workers in $\Gamma$, where $\mathbf{u}^j\in\{0,1\}^W$ for all $j\in\{1,\ldots,L\}$ and $\mathbf{0}$ is the empty set of workers. Before proceeding to the formal definition of $\widehat{Ch}_f$, we want to give an intuitive description of this choice function, which resembles that of the Probabilistic Serial assignment of \cite{BM01}. Consider the workers of each worker type of $\mathbf{x}$ as a ``divisible commodity". The firm first consumes the workers of each type of $\mathbf{u}^j$ from $\mathbf{x}$ simultaneously at speed 1, where $j$ is the smallest index such that $\mathbf{u}^j$ is available at $\mathbf{x}$. When the workers of one worker type of $\mathbf{u}^j$ is exhausted, the firm switches to consume the workers of $\mathbf{u}^{j'}$ simultaneously at speed 1, where $j'$ is the smallest index such that $\mathbf{u}^{j'}$ is available at the remaining subpopulation. This procedure goes on until the time reaches 1, or until there is no acceptable set of worker types available at the remaining subpopulation.

\begin{example}
\normalfont
If $f$ has preference $\{w_1,w_2\}\succ\{w_2,w_3\}\succ\{w_3\}\succ\emptyset$ in $\Gamma$, then $\widehat{Ch}_f((0.6,0.6,0.5))=(0.6,0.6,0.4)$: It first consumes 0.6 of each of the type-$w_1$ and type-$w_2$ workers. Then, since $\{w_2,w_3\}$ is not available at the remaining subpopulation (0,0,0.5), it switches to consume 0.4 of the type-$w_3$ workers when the time reaches 1. In another case, $\widehat{Ch}_f((0.1,0.4,0.1))=(0.1,0.2,0.1)$: The firm first consumes 0.1 of each of the type-$w_1$ and type-$w_2$ workers, and then, it consumes 0.1 of each of the type-$w_2$ and type-$w_3$ workers from the remaining subpopulation (0,0.3,0.1). The procedure then terminates since there is no acceptable set of worker types available at the remaining subpopulation (0,0.2,0).
\end{example}

Formally, for any $\mathbf{x}\in[0,1]^W$, $\widehat{Ch}_f(\mathbf{x})$ is defined by the following recursive procedure.
Let $t_0=0$, $\mathbf{z}^0=\mathbf{x}$. Suppose $t_0$, $\mathbf{z}^0$, \ldots, $t_{k-1}$, $\mathbf{z}^{k-1}$ $(k\in\{1,2,\ldots,L\})$ are already defined, then define
\begin{equation}\label{proce}
\begin{aligned}
&t_k=\min\{1-\sum^{k-1}_{j=0}t_j, z^{k-1}_i\mid i\in\{1,2,\ldots,n\}\quad\text{and}\quad u^k_i\neq 0\}\\
&\mathbf{z}^k=\mathbf{z}^{k-1}-t_k\mathbf{u}^k
\end{aligned}
\end{equation}
After the sequences $(t_k)_{k=0}^L$ and $(\mathbf{z}^k)_{k=0}^L$ have been defined, we define
\begin{equation}\label{choice}
\widehat{Ch}_f(\mathbf{x})=\sum_{k=1}^Lt_k\mathbf{u}^k
\end{equation}
When we consider this procedure as the consumption process described above, for each $k\in\{1,2,\ldots,L\}$, $t_k$ is the time the firm spends in the consumption of $\mathbf{u}^k$; $\mathbf{z}^k$ is the remaining subpopulation after $\mathbf{u}^k$ has been consumed. If $\sum^{k'}_{j=1}t_j=1$ for some $k'\in\{1,2,\ldots,L-1\}$, then $t_{k''}=0$ for all $k''\in\{k'+1,\ldots,L\}$.
It turns out that the choice function $\widehat{Ch}_f$ constructed above is continuous and satisfies the revealed preference property.

\begin{lemma}\label{lma_cont}
\normalfont
$\widehat{Ch}_f$ is continuous and satisfies the \textbf{revealed preference property}: For any $\mathbf{x},\mathbf{x}'\in [0,1]^W$ with $\mathbf{x}'\leq \mathbf{x}$, $\widehat{Ch}_f(\mathbf{x})\leq \mathbf{x}'$ implies $\widehat{Ch}_f(\mathbf{x}')=\widehat{Ch}_f(\mathbf{x})$.
\end{lemma}

The revealed preference property stated above is known as Sen's property $\alpha$, which is a premise of CKK's existence theorem. Let $\widehat{Ch}_F$ be the set of choice functions of all firms constructed in the above way. A continuum market induced from a matching market $\Gamma$ is also summarized as a tuple $\widehat{\Gamma}=(W,F,\succ_W,\widehat{Ch}_F)$, where $W, F$ and $\succ_W$ are the same as those of $\Gamma$. We can interpret the induced market $\widehat{\Gamma}$ as a schedule matching market of \cite{AG03}, where each worker schedules her time among different firms, and each firm schedules its time among different groups of workers. Different groups of workers bring different outputs per unit time for each firm, which decrease along with the firm's preference order in the original market $\Gamma$. For instance, suppose $f$ has preference $\{w_1,w_2\}\succ\{w_2,w_3\}\succ\{w_3\}\succ\emptyset$ in $\Gamma$. In the induced market, $\{w_1,w_2\}$ brings more output per unit time than $\{w_2,w_3\}$, which brings more output per unit time than $\{w_3\}$. When a firm faces an available supply of working time from workers, the firm solves a linear program to find the optimal schedule, which is the same as the choice function $\widehat{Ch}_f$.\footnote{I thank coeditor Federico Echenique for this insightful interpretation.}

A \textbf{matching} $M$ in $\widehat{\Gamma}$ assigns each firm a subpopulation of workers: $M=(M_f)_{f\in{\widetilde{F}}}$ such that $M_f\in [0,1]^W$ for all $f\in \widetilde{F}$ and $\sum_{f\in{\widetilde{F}}}M_f(w)=1$ for all $w\in W$. For any subpopulations $\mathbf{x}$,$\mathbf{x}'\in[0,1]^W$, we let $\mathbf{x}\vee \mathbf{x}'$ denote the subpopulation whose quantity of type-$w$ workers is $max\{x(w),x'(w)\}$. We use firms' choice functions to define firms' preferences over matchings. Given two matchings $M$ and $M'$, we say that firm $f$ prefers $M'_f$ to $M_f$, denoted as $M'_f\succeq_f M_f$, if $M'_f=\widehat{Ch}_f(M'_f\vee M_f)$. This relation is known as Blair's partial order in the literature (\citealp{B84}). We write $M'_f\succ_f M_f$ to indicate that $M'_f\succeq_f M_f$ and $M'_f\neq M_f$. For any matching $M$ and firm $f$, the subpopulation of workers assigned to firm $f$ or some firm worse than $f$ (according to their preferences) in $M$ is denoted as $A^{\preceq f}(M)\in [0,1]^W$, where
\begin{equation}\label{DfM}
A^{\preceq f}(M)(w)=\sum_{f'\in \widetilde{F}:f\succeq_w f'}M_{f'}(w)
\end{equation}
for each $w\in W$. $A^{\preceq f}(M)$ refers to the available subpopulation for $f$ in $M$ since it measures the amount of workers of each type that would rather match with $f$ in $M$. The concept of stability in the continuum market $\widetilde{\Gamma}$ is then defined as follows:

\begin{definition}\label{stable}
\normalfont
A matching $M=(M_f)_{f\in{\widetilde{F}}}$ in $\widehat{\Gamma}$ is \textbf{stable} if
\begin{description}
\item[(i)] (Individual Rationality) For each $f\in F$, $M_f= \widehat{Ch}_f(M_f)$; and for all $w\in W$, $M_f(w)=0$ for any $f$ that satisfies ${\o}\succ_w f$.

\item[(ii)] (No Blocking Coalition) There are no $f\in F$ and $M'_f\in [0,1]^W$ such that $M'_f\succ_f M_f$ and $M'_f\leq A^{\preceq f}(M)$.
\end{description}
\end{definition}

CKK proved that stable matchings exist when each firm has a continuous choice function.\footnote{See Theorem 2 of CKK. CKK's framework accommodates  indifferences in firms' preferences. The requirements of continuity and convex-valuedness on firms' choice correspondences in their setting reduce to the continuity of firms' choice functions in our setting.} Then by Lemma \ref{lma_cont}, we know that there exists a stable matching in the continuum market $\widehat{\Gamma}$.

\begin{example}\label{exam_continuum}
\normalfont
Consider a matching market $\Gamma$ where there are two firms $F=\{f_1,f_2\}$ and three workers $W=\{w_1,w_2,w_3\}$. The preferences of firms and workers are as follows.
\begin{align*}
&f_1: \{w_1,w_2\}\succ\{w_3\}\succ\emptyset \quad &w_1: &\quad f_1\succ f_2\succ{\o}\\
&f_2: \{w_1,w_2\}\succ\emptyset \quad &w_2: &\quad f_2\succ f_1\succ{\o}\\
&\quad &w_3: &\quad f_1\succ{\o}
\end{align*}

\end{example}

The above matching market induces a continuum market $\widehat{\Gamma}$, where the choice functions $\widehat{Ch}_{f_1}$ and $\widehat{Ch}_{f_2}$ are generated from the above firms' preference profile via (\ref{proce}) and (\ref{choice}). By Lemma \ref{lma_cont} and CKK's existence theorem, there is at least one stable matching in this continuum market. For instance, the matching $M$ in $\widehat{\Gamma}$, where $f_1$ is matched with $(0.5,0.5,0.5)$, $f_2$ with $(0.5,0.5,0)$, and ${\o}$ with $(0,0,0.5)$, is stable. Under our constructed firms' preferences, $\widehat{Ch}_{f_1}((0.5,0.5,0.5))=(0.5,0.5,0.5)$ and $\widehat{Ch}_{f_2}((0.5,0.5,0))=(0.5,0.5,0)$ indicate the individual rationality of $M$. Hiring more type-$w_3$ workers does not benefit $f_1$. $f_1$ and $f_2$ would both like to hire more workers of type-$w_1$ and type-$w_2$ at the ratio 1:1, but neither can draw workers of this ratio from the other. For example, $f_1$ would be better off when matched with $(0.6,0.6,0.4)$. This is because $\widehat{Ch}_{f_1}((0.6,0.6,0.4)\vee(0.5,0.5,0.5))=(0.6,0.6,0.4)$ indicates $(0.6,0.6,0.4)\succ_{f_1}(0.5,0.5,0.5)$. However, $f_1$ cannot draw any type-$w_2$ workers from $f_2$ since the type-$w_2$ workers prefer $f_2$ to $f_1$. In the language of Definition \ref{stable}, $f_1$ and $(0.6,0.6,0.4)$ do not form a blocking coalition because although $(0.6,0.6,0.4)\succ_{f_1}(0.5,0.5,0.5)$ holds, $(0.6,0.6,0.4)\leq A^{\preceq f_1}(M)$ does not hold where $A^{\preceq f_1}(M)=(1,0.5,1)$.

We say that a matching $M$ in $\widehat{\Gamma}$ is integral if $M_f(w)\in\{0,1\}$ for all $f\in \widetilde{F}$ and $w\in W$, otherwise we say that $M$ is fractional. Another observation is that each stable integral matching in the continuum market $\widehat{\Gamma}$ is also a stable matching in the original market $\Gamma$. For instance, the stable integral matching in $\widehat{\Gamma}$ where $f_1$ matches with (0,0,0), $f_2$ with (1,1,0), and ${\o}$ with $(0,0,1)$ is also a stable matching in $\Gamma$.

\begin{lemma}\label{lma_between}
\normalfont
If $M$ is a stable integral matching in the continuum market $\widehat{\Gamma}$, then $\mu$ is a stable matching in the original matching market $\Gamma$, where $\mu(w)=f$  and $w\in\mu(f)$ if $M_f(w)=1$ for each $w\in W$ and $f\in \widetilde{F}$.
\end{lemma}

\subsection{Stability-preserving turnovers}\label{Sec_preserve}

Every matching market $\Gamma$ induces a continuum market $\widehat{\Gamma}$. Motivated by Lemma \ref{lma_between}, we turn to investigate when there exists a stable integral matching in $\widehat{\Gamma}$. According to CKK's existence theorem, there always exists a stable matching $M$ in $\widehat{\Gamma}$, which may be fractional or integral. Consider the stable fractional matching in Example \ref{exam_continuum}, where $f_1$ is matched with $(0.5,0.5,0.5)$, $f_2$ with $(0.5,0.5,0)$, and ${\o}$ with $(0,0,0.5)$. We obtain a stable integral matching when all the workers matched with $f_2$ switch to $f_1$ and the type-$w_3$ workers previously hired by $f_1$ become unemployed. Such ``turnover" of workers in the continuum market preserves stability and produces a stable integral matching. However, consider the continuum market induced by (\ref{exam_in1}) and the stable fractional matching where $f_1$ is matched with $(0.5,0.5)$, and $f_2$ with $(0.5,0.5)$. We cannot obtain a stable integral matching by any stability-preserving turnover of workers.\footnote{It happens that there is no stability-preserving turnovers of workers in this continuum market, because there is only one stable matching; see Example 3 of CKK.}

Therefore, a stable matching always exists in the original market when  stability-preserving turnover toward a stable integral matching exists in the continuum market. To formalize some of the stability-preserving turnovers, we generalize the concept of matching to ``pseudo-matching" in the continuum market and introduce ``stable transformations" that operate on pseudo-matchings in a succinct manner. We find that some special stability-preserving turnovers toward an integral matching can be expressed as a series of our stable transformations.

In the continuum market $\widehat{\Gamma}$, we call $M=(M_f)_{f\in{\widetilde{F}}}$ a \textbf{pseudo-matching} if $M_f\in [0,1]^W$ for all $f\in \widetilde{F}$. A pseudo-matching may assign more or less than quantity 1 of some worker type to firms. A pseudo-matching is a matching if $\sum_{f\in{\widetilde{F}}}M_f(w)=1$ for each $w\in W$. We say that a pseudo-matching $M$ is \textbf{stable} if it satisfies condition (i) and (ii) in Definition \ref{stable}. In other words, a pseudo-matching $M$ is stable if $M$ is a stable matching when the quantity of each worker type $w$ in the market is adjusted to its current quantity $\sum_{f\in{\widetilde{F}}}M_f(w)$ in $M$.

\begin{example}\label{exam_pseudo}
\normalfont
Consider the continuum market $\widehat{\Gamma}$ in Example \ref{exam_continuum}. Let $M$ be the pseudo-matching where $M_{f_1}=(0.6,0.6,0.3)$, $M_{f_2}=(0.3,0.3,0)$, and $M_{{\o}}=(0,0,0)$, then $M$ is a stable pseudo-matching. The reason is similar to that in Example \ref{exam_continuum}.

Let $M'$ be the pseudo-matching where $M'_{f_1}=(0,0,0)$, $M'_{f_2}=(0.3,0.3,0)$ and $M'_{{\o}}=(0,0,0)$. Readers can check that $M'$ is also a stable pseudo-matching.
\end{example}

Given a stable pseudo-matching $M$ in $\widehat{\Gamma}$, each of the following transformations on $M$ produces a stable pseudo-matching $M'$.

\textbf{Type-1 stable transformation}: Choose a firm $f'$ from $F$ such that $\sum_{j=1}^Lt_j<1$ holds in the procedure (\ref{proce}) that computes $\widehat{Ch}_{f'}(M_{f'})$. Let $M'_{f'}=\mathbf{0}$ and let $M'_{f}=M_f$ for all $f\neq f'$.

\textbf{Type-2 stable transformation}: Choose a firm $f'$ from $F$. Consider the procedure (\ref{proce}) that computes $\widehat{Ch}_{f'}(M_{f'})=\sum_{j=1}^Lt_j\mathbf{u}^j$. Choose an index $k\in\{1,2,\ldots,L\}$ that satisfies $t_k>0$, let $M'_{f'}=\mathbf{u}^k$. Let $M'_{f}=M_f$ for all $f\neq f'$.

\textbf{Type-3 stable transformation}: Choose a worker type $w'\in W$ that satisfies $M_{{\o}}(w')\in(0,1)$; let $M'_{{\o}}(w')=0$ or $1$. Let $M'_{{\o}}(w)=M_{{\o}}(w)$ for all $w\neq w'$, and $M'_{f}=M_f$ for all $f\in F$.

\begin{example}
\normalfont
Consider the continuum market $\widehat{\Gamma}$ in Example \ref{exam_continuum} and the stable pseudo-matchings $M$ and $M'$ in Example \ref{exam_pseudo}. $M\rightarrow M'$ is a type-1 stable transformation. Note that in the procedure (\ref{proce}) that computes $\widehat{Ch}_{f_1}(M_{f_1})$, we have $\sum_{j=1}^Lt_j=0.9<1$. Since $f_1$ can not draw any subpopulation that includes both types from $f_2$ in $M$, $f_1$ can do no such thing in $M'$ either, and thus, this transformation maintains stability.

Let $M''$ be the pseudo-matching where $M''_{f_1}=(1,1,0)$, $M''_{f_2}=(0.3,0.3,0)$, and $M''_{{\o}}=(0,0,0)$; then, $M\rightarrow M''$ is a type-2 stable transformation. This transformation maintains stability because $f_2$ can not draw any subpopulation that includes both types from $f_1$ no matter when $f_1$ is matched with $(0.6,0.6,0.3)$ or $(1,1,0)$.

Let $\widetilde{M}$ be the pseudo-matching where $\widetilde{M}_{f_1}=(0,0,0)$, $\widetilde{M}_{f_2}=(0.3,0.3,0)$, and $\widetilde{M}_{{\o}}=(0,0.3,0)$. Let $\widetilde{M}'$ be the pseudo-matching where $\widetilde{M}'_{f_1}=(0,0,0)$, $\widetilde{M}'_{f_2}=(0.3,0.3,0)$ and $\widetilde{M}'_{{\o}}=(0,1,0)$. Then, $\widetilde{M}\rightarrow M'$ and $\widetilde{M}\rightarrow \widetilde{M}'$ are both type-3 stable transformations. Unmatched workers in a stable matching can be viewed as redundant for firms; thus, stability is not affected when the quantity of unmatched workers varies.
\end{example}

\begin{lemma}\label{lma_trans}
\normalfont
Each of the type-1, type-2, and type-3 stable transformations on a stable pseudo-matching produces a stable pseudo-matching.
\end{lemma}

When we implement several stable transformations on a stable fractional matching, the matching becomes ``less fractional" and ultimately becomes integral, but may it assign more or less than quantity 1 of some worker types. Thus, we also want a series of transformations to be ``balanced", that is, to assign 1 quantity of each worker type in the output. Because the stable transformations transform stable fractional matchings in a succinct manner, we can reduce ``balanced" stable transformations to integral solutions to a system of linear equations, where the total unimodularity condition applies.

\subsection{Illustrative example}\label{Sec_illu}

\begin{example}\label{illustrate}
\normalfont
Consider the market of Example \ref{exam_continuum} and the stable matching $M$ in the continuum market where
$M_{f_1}=(0.5,0.5,0.5)$, $M_{f_2}=(0.5,0.5,0)$, and $M_{{\o}}=(0,0,0.5)$. In the following, we represent a pseudo-matching in $\widehat{\Gamma}$ with a matrix, the rows of which represent the subpopulations matched with the firms. $M$ is then represented by the following $3\times3$ matrix.

\begin{center}
\begin{tabular}
[c]{c|ccc}
& $w_1$ & $w_2$ & $w_3$\\\hline
$f_1$ & $0.5$ & $0.5$ & $0.5$\\
$f_2$ & $0.5$ & $0.5$ & $0$\\
${\o}$ & $0$ & $0$ & $0.5$%
\end{tabular}
\end{center}

Consider the following transformations on $M$.\\

\begin{tikzpicture}[baseline = (M.west)]
    \tikzset{brace/.style = {decorate, decoration = {brace, amplitude = 5pt}, thick}}
    \matrix(M)
    [
        matrix of math nodes,
        left delimiter = (,
        right delimiter = )
    ]
    {
0.5 & 0.5 & 0.5\\
0.5 & 0.5 & 0\\
0 & 0 & 0.5\\
    };
\end{tikzpicture}
$\stackrel{(1)}{\Longrightarrow}$
\begin{tikzpicture}[baseline = (M.west)]
    \tikzset{brace/.style = {decorate, decoration = {brace, amplitude = 5pt}, thick}}
    \matrix(M)
    [
        matrix of math nodes,
        left delimiter = (,
        right delimiter = )
    ]
    {
0.5 & 0.5 & 0.5\\
0 & 0 & 0\\
0 & 0 & 0.5\\
    };
\end{tikzpicture}
$\stackrel{(2)}{\Longrightarrow}$
\begin{tikzpicture}[baseline = (M.west)]
    \tikzset{brace/.style = {decorate, decoration = {brace, amplitude = 5pt}, thick}}
    \matrix(M)
    [
        matrix of math nodes,
        left delimiter = (,
        right delimiter = )
    ]
    {
1 & 1 & 0\\
0 & 0 & 0\\
0 & 0 & 0.5\\
    };
\end{tikzpicture}

$\stackrel{(3)}{\Longrightarrow}$
\begin{tikzpicture}[baseline = (M.west)]
    \tikzset{brace/.style = {decorate, decoration = {brace, amplitude = 5pt}, thick}}
    \matrix(M)
    [
        matrix of math nodes,
        left delimiter = (,
        right delimiter = )
    ]
    {
1 & 1 & 0\\
0 & 0 & 0\\
0 & 0 & 1\\
    };
\end{tikzpicture}\\

Transformations (1), (2), and (3) are type-1, type-2, and type-3 stable transformations, respectively. It turns out that we finally reach a stable integral matching $M'$ in $\widehat{\Gamma}$, where $M'_{f_1}=(1,1,0)$, $M'_{f_2}=(0,0,0)$, and $M'_{{\o}}=(0,0,1)$. The outcome is not only a stable integral pseudo-matching but also a matching that assigns precisely quantity 1 of each worker type.

A critical observation is that we can reach a stable integral matching through stable transformations on $M$ when there is a nonnegative integral solution to the following system of linear equations.
\medskip

\begin{tikzpicture}[baseline = (M.west)]
    \tikzset{brace/.style = {decorate, decoration = {brace, amplitude = 5pt}, thick}}
    \matrix(M)
    [
        matrix of math nodes,
        left delimiter = (,
        right delimiter = )
    ]
    {
        1 & 1 & 0 & 0 & 0\\
        0 & 0 & 1 & 1 & 0\\
        1 & 0 & 1 & 0 & 0\\
        1 & 0 & 1 & 0 & 0\\
        0 & 1 & 0 & 0 & 1\\
            };
     \draw[decorate,decoration={brace,mirror,amplitude=3mm},thick]
        ($(M-3-1.north east) + (-1, 0)$)
        -- node[left = 18pt]{$B^*$}
        ($(M-5-1.south east) + (-1, 0)$);
\end{tikzpicture}
\begin{tikzpicture}[baseline = (M.west)]
    \tikzset{brace/.style = {decorate, decoration = {brace, amplitude = 5pt}, thick}}
        \matrix(M)
    [
        matrix of math nodes,
        left delimiter = (,
        right delimiter = )
    ]
    {
       z_1\\
z_2\\
z_3\\
z_4\\
z_5\\
    };
\end{tikzpicture}$=$
\begin{tikzpicture}[baseline = (M.west)]
    \tikzset{brace/.style = {decorate, decoration = {brace, amplitude = 5pt}, thick}}
        \matrix(M)
    [
        matrix of math nodes,
        left delimiter = (,
        right delimiter = )
    ]
    {
1\\
1\\
1\\
1\\
1\\
    };
\end{tikzpicture}

\medskip
Let $B$ denote the $5\times 5$ matrix on the left. The first and second rows of $B$ are constraints for preserving stability; the third to fifth rows of $B$ are constraints for assigning quantity 1 of each worker type. Let $B^*$ denote the $3\times 5$ submatrix that includes the third to fifth rows of $B$. Let $B_i$ and $B^*_i$ be the vectors of the $i$-th column of $B$ and $B^*$, respectively. $B^*_1$ and $B^*_2$ correspond to $\{w_1,w_2\}$ and $\{w_3\}$ in $f_1$'s preference list, respectively. $B^*_3$ and $B^*_4$ correspond to $\{w_1,w_2\}$ and $\emptyset$ in $f_2$'s preference list, respectively. $B^*_5$ corresponds to the type-$w_3$ workers matched with firm ${\o}$. $\mathbf{z}=(0.5,0.5,0.5,0.5,0.5)$ is a solution to this system, which refers to matching $M$ as follows.

\begin{equation}\label{z-matching}
M=\left(
             \begin{aligned}
             f_1  \qquad\qquad & \qquad\qquad f_2 \qquad & \qquad {\o} \\
             z_1(1,1,0)+z_2(0,0,1)\quad&\quad z_3(1,1,0)+z_4(0,0,0)\quad&\quad z_5(0,0,1)
             \end{aligned}
\right)
\end{equation}

$(z_1,z_2)=(0.5,0.5)$ corresponds to $(t_1,t_2)$ in the procedure (\ref{proce}) that computes $\widehat{Ch}_{f_1}(0.5,0.5,0.5)$. $(z_3,z_4)=(0.5,0.5)$, corresponds to $(t_1,1-t_1)$ in the procedure (\ref{proce}) that computes $\widehat{Ch}_{f_2}(0.5,0.5,0)$. $z_5=0.5$ corresponds to $M_{{\o}}=(0,0,0.5)$. Note that (\ref{z-matching}) with any integral $\mathbf{z}\in \{0,1\}^5$ that satisfies $z_1+z_2=1$ and $z_3+z_4=1$ (guaranteed by the first and second rows of $B$) corresponds to a stable integral pseudo-matching obtained from $M$ via stable transformations. Then, any $\mathbf{z}\in \{0,1\}^5$ that satisfies $B\mathbf{z}=\mathbf{1}$ corresponds to a stable integral matching since $B^*\mathbf{z}=\mathbf{1}$ requires that the workers of each type assigned to firms is of quantity 1. For instance, $\mathbf{z}'=(1,0,0,1,1)$ is a solution to the system, and we obtain the stable integral matching $M'$ by plugging $\mathbf{z}'$ into (\ref{z-matching}).

Therefore, given a stable matching $M$ in the continuum market $\widehat{\Gamma}$, we can construct a system of linear equations $B\mathbf{z}=\mathbf{1}$. Our construction guarantees that the polytope $\{\mathbf{z}\mid B\mathbf{z}=\mathbf{1}, \mathbf{z}\geq0\}$ is nonempty because $M$ corresponds to a nonnegative solution to this system of equations. Our construction also guarantees that every integral point of this polytope corresponds to a stable matching in the original market. Now, we can apply a standard result from integer programming to this problem. Under the condition that matrix $B$ is unimodular, all vertices of the polytope $\{\mathbf{z}\mid B\mathbf{z}=\mathbf{1}, \mathbf{z}\geq0\}$ are integral (\citealp{HK56}; see also Theorem 21.5 of \citealp{S86}). Since the polytope is nonempty, we further know that there is at least an integral vertex on this polytope.

Finally, it is not difficult to find out that $B$ is unimodular when the firms' demand type $\mathcal{D}$ is totally unimodular. In this example, we have $\mathcal{D}_{f_1}=\{(1,1,0),(1,1,-1),(0,0,1)\}$, $\mathcal{D}_{f_2}=\{(1,1,0)\}$. The firms' demand type is $\mathcal{D}=\{(1,1,0),(1,1,-1),(0,0,1)\}$, which is totally unimodular. Now, we illustrate why total unimodularity of $\mathcal{D}$ implies unimodularity of matrix $B$. For instance, $B_1,B_2,B_3$, and $B_5$ are linearly independent. The total unimodularity of $\mathcal{D}$ guarantees that matrix $[B^*_1-B^*_2, B^*_5]$ is unimodular,\footnote{This is because of the following two facts: (i) Matrix $H$ is totally unimodular if and only if $[H,I]$ is unimodular ($I$ is the identity matrix). (ii) Any column submatrix of a unimodular matrix is unimodular again. Fact (i) implies the unimodularity of $[\mathcal{D},I]$. Thus, fact (ii) implies the unimodularity of $[B^*_1-B^*_2, B^*_5]$ since $B^*_1-B^*_2$ is an element of $\mathcal{D}$ and $B^*_5$ is a column from $I$.\label{footnote}} and thus the set $\{B^*_1-B^*_2, B^*_5\}$ can be extended to a basis for $\mathbb{R}^3$, of integer vectors, with determinant $\pm1$ (e.g., by extending with $(1,0,0)$). Then, we know that $\{B_1, B_3\}\cup\{B_1-B_2, B_5\}$ can be extended to a basis for $\mathbb{R}^5$, of integer vectors, with determinant $\pm1$ (e.g., by extending with $(0,0,1,0,0)$). Therefore, such extension also exists for the set $\{B_1, B_2, B_3, B_5\}$.\footnote{Since subtracting one column from another column leaves the determinant unchanged, we know that $\{B_1, B_2, B_3, B_5\}\cup\{(0,0,1,0,0)\}$ is also an integral basis for $\mathbb{R}^5$, with determinant $\pm1$.}
\end{example}

\subsection{Relation to the LP method}\label{Sec_LP}

Linear programming has been used to select stable matchings of specific properties, such as optimal marriage and median stable matching. Both the LP method and our method represent matchings by their indicator vectors and study polytopes related to these vectors. We discuss the differences between the two methods below.

First, the LP method is used to find stable matchings of specific properties in a discrete matching market when firms have so-called \emph{responsive} preferences, a special case of substitutable preferences. By contrast, we investigated under what conditions stable matching exist in the market with complementarities. Second, the LP method uses the definition of stable matching to define a polytope whose vertices are the stable matchings of the market; then, the problems of optimal marriage and median stable matching reduce to linear programs. By contrast, we address our problem by constructing a fictitious market where each worker is divisible and investigating the stable integral matchings of this market. We incorporate the results of CKK and tools from integer programming to obtain our results.

\section{Application\label{Sec_App}}

This section presents an application of matching firms with specialists, which has a practical economic meaning. We describe the structure of firms' acceptable sets of workers as a directed rooted tree, which we call the ``technology tree". Each vertex of the technology tree represents a technology that requires a set of workers to implement. Each edge of the tree is an upgrade from one technology to another that requires more workers. A worker is called a specialist if she engages in only one upgrade. We show that firms' demand type is totally unimodular when firms have unit-demand preferences over the technologies of a technology tree where each worker is a specialist. We first provide an illustrative example as follows.

\begin{example}\label{exam_app}
\normalfont
Consider a market with two firms $f_1$, $f_2$, and four workers $w_1, w_2, w_3, w_4$. A technology tree is depicted as follows.

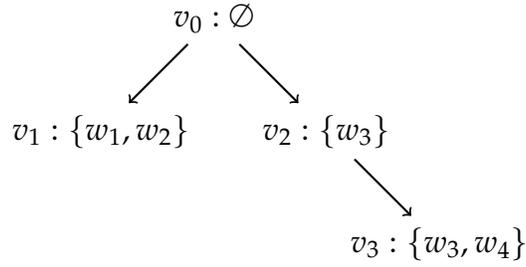
\begin{center}
\begin{tikzpicture}[thick,->]
	\node {$v_0:\emptyset$}
	child {node {$v_1:\{w_1,w_2\}$}}
   	child [missing] {}	
	child { node {$v_2:\{w_3\}$}
		child [missing] {}
		child [missing] {}
		child {node {$v_3:\{w_3,w_4\}$}
			child [missing] {}
		}
	};
	\end{tikzpicture}
\end{center}

Each vertex from $\{v_0,v_1,v_2,v_3\}$ represents a technology that requires the set of workers on the right to implement. The root $v_0$ represents no technology and requires no worker. Each directed edge is an upgrade from one technology to another, where more workers should be employed to implement the upgrade. If $e=vv'$ is an edge from vertex $v$ to vertex $v'$, where $w$ is not demanded by $v$ but demanded by $v'$, we say that $w$ engages in the upgrade $e$ or $vv'$. For example, $w_1$ and $w_2$ both engage in the upgrade $v_0v_1$. Each firm possesses some of the technologies and has a preference order over the technologies it possesses, which induces its preference over the sets of workers required for the technologies. For instance, $f_1$ may possess $v_1$ and $v_2$; $f_2$ may possess $v_1$ and $v_3$; and $v_0$ is trivially possessed by both firms. $f_1$ and $f_2$ may have the following preferences.

\begin{equation}\label{pre_app}
\begin{aligned}
&f_1: \{w_1,w_2\}\succ \{w_3\}\succ \emptyset\\
&f_2: \{w_3,w_4\}\succ \{w_1,w_2\}\succ \emptyset
\end{aligned}
\end{equation}
\end{example}

With firms' preferences induced in this way, we find that firms' demand type is totally unimodular when the technology tree satisfies the following condition.

\begin{equation}\label{special_word}
\text{\emph{Each worker is a specialist that engages in only one upgrade.}}
\end{equation}

For example, given that $w_1$ engages in the upgrade $v_0v_1$, this condition requires that $w_1$ cannot engage in other upgrades. When firms' acceptable sets of workers are from a technology tree that satisfies this condition, the firms' demand type is totally unimodular since its elements form a \emph{network matrix} (\citealp{T65}, see also Chapter 19.3 of \citealp{S86}). See an illustration in Section \ref{proof_special}. To see the role of (\ref{special_word}), consider the firms' preference profile in (\ref{exam_in1}). There are two possible structures for the technology tree that induces the firms' preferences.
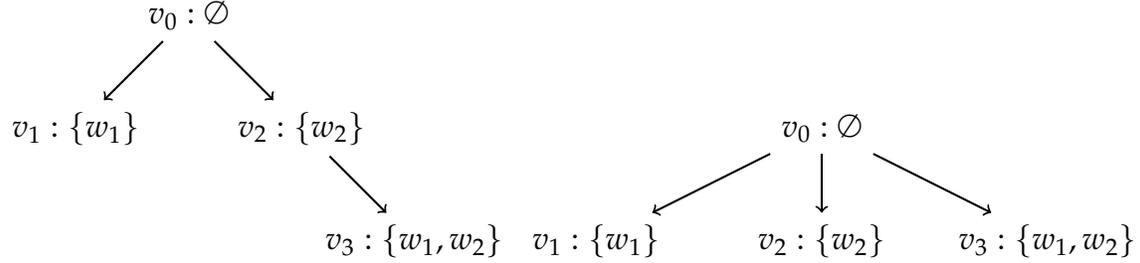
\begin{center}
\begin{tikzpicture}[thick,->]
	\node {$v_0:\emptyset$}
	child {node {$v_1:\{w_1\}$}}
   	child [missing] {}	
	child { node {$v_2:\{w_2\}$}
		child [missing] {}
		child [missing] {}
		child {node {$v_3:\{w_1,w_2\}$}
			child [missing] {}
		}
	};
\end{tikzpicture}
\begin{tikzpicture}[thick,->]
	\node {$v_0:\emptyset$}
	child {node {$v_1:\{w_1\}$}}
child [missing] {}
   	child { node {$v_2:\{w_2\}$}}
   child [missing] {}
    child { node {$v_3:\{w_1,w_2\}$}};
\end{tikzpicture}
\end{center}
Both technology trees violate condition (\ref{special_word}): On the left, $w_1$ engages in both upgrades $v_0v_1$ and $v_2v_3$. On the right, $w_1$ engages in both upgrades $v_0v_1$ and $v_0v_3$.

Formally, a \textbf{technology tree} $T=(V,E,W)$ is a directed rooted tree $(V,E)$ defined on a set of workers $W$. $V=\{v_0,v_1,\ldots,v_l\}$ is a set of vertices with $v_0$ as the root. Each vertex $v\in V$ represents a technology that requires a subset of workers $W^v\subseteq W$ to implement. The root $v_0$ represents no technology and requires no worker: $W^{v_0}=\emptyset$. $E$ is a set of directed edges, all of which point away from the root. For each edge $e\in E$ from vertex $v$ to vertex $v'$, $W^v\subset W^{v'}$, and we let $W^e=W^{v'}\setminus W^{v}$. We say that worker $w$ is a \textbf{specialist} in $T$ if she engages in only one upgrade: $\mid\{e\in E\mid w\in W^e\}\mid=1$. We study firms' preferences where each firm wants to hire one set of workers from a common technology tree.

\begin{definition}
\normalfont
Firms have unit-demand preferences over a technology tree $T=(V,E,W)$ if for each $f\in F$ and each $S\subseteq W$, $S\succ_f \emptyset$ implies $S=W^v$ for some $v\in V$.
\end{definition}

In words, firms have unit-demand preferences over a technology tree if those sets of workers not on the tree are not acceptable for all firms. Each firm wants to employ one set from the technology tree and can have arbitrary preference order over the sets of workers on the tree.\footnote{Note that if a firm has preference $\{w_3\}\succ \{w_3,w_4\}\succ \emptyset$ in Example \ref{exam_app}, this preference is essentially equivalent to $\{w_3\}\succ\emptyset$.  } We then have the following theorem.

\begin{theorem}\label{thm_special}
\normalfont
Firms' demand type is totally unimodular if firms have unit-demand preferences over a technology tree where each worker is a specialist.
\end{theorem}

Therefore, we know that a stable matching always exists when firms have unit-demand preferences over a technology tree where each worker is a specialist. Studying this problem motivated the author's subsequent work, \cite{H21b}, which generalizes condition (\ref{special_word}) and the unit demand of technologies. Both generalizations allow for firms' demand types that are not totally unimodular.

\section{Appendix}

\subsection{Proofs of the lemmata}

\textit{Proof of Lemma \ref{lma_cont}} \quad

(1) Revealed preference property. For any $\mathbf{x},\mathbf{\widetilde{x}}\in [0,1]^W$ with $\mathbf{\widetilde{x}}\leq \mathbf{x}$ and $\widehat{Ch}_f(\mathbf{x})\leq \mathbf{\widetilde{x}}$, consider the procedures (\ref{proce}) that compute $\widehat{Ch}_f(\mathbf{x})$ and $\widehat{Ch}_f(\mathbf{\widetilde{x}})$. Let $(t_k)_{k=1}^L$ and $(\mathbf{z}^k)_{k=1}^L$ be the parameters in computing $\widehat{Ch}_f(\mathbf{x})$. Let $(\widetilde{t}_k)_{k=1}^{L}$,  and $(\mathbf{\widetilde{z}}^k)_{k=1}^{L}$ be the parameters in computing $\widehat{Ch}_f(\mathbf{\widetilde{x}})$.

Since $\mathbf{\widetilde{x}}\leq \mathbf{x}$, we have $\widetilde{t}_1\leq t_1$. Suppose $\widetilde{t}_1< t_1$, we have $\widehat{Ch}_f(\mathbf{x})\geq t_1\mathbf{u}^1$ but $\mathbf{\widetilde{x}}\geq t_1\mathbf{u}^1$ does not hold. This contradicts $\widehat{Ch}_f(\mathbf{x})\leq \mathbf{\widetilde{x}}$. Thus, we have $t_1=\widetilde{t}_1$, and then $\mathbf{\widetilde{z}}^1\leq\mathbf{z}^1$.

Suppose for all $j\in\{1,2,\ldots,k-1\}$, we have $t_j=\widetilde{t}_j$ and $\mathbf{z}^j\geq\mathbf{\widetilde{z}}^j$. Then, since $\mathbf{z}^{k-1}\geq\mathbf{\widetilde{z}}^{k-1}$, we have $\widetilde{t}_k\leq t_k$. Suppose $\widetilde{t}_k< t_k$, we have $\widehat{Ch}_f(\mathbf{x})\geq \sum_{j=1}^kt_j\mathbf{u}^j$ but $\mathbf{\widetilde{x}}\geq \sum_{j=1}^kt_j\mathbf{u}^j$ does not hold. This contradicts $\widehat{Ch}_f(\mathbf{x})\leq \mathbf{\widetilde{x}}$. Thus, we have $t_k=\widetilde{t}_k$ and $\mathbf{\widetilde{z}}^k\leq\mathbf{z}^k$.

According to the above inductive arguments, the procedures that compute $\widehat{Ch}_f(\mathbf{x})$ and $\widehat{Ch}_f(\mathbf{\widetilde{x}})$ coincide, and we have $\widehat{Ch}_f(\mathbf{x})= \widehat{Ch}_f(\mathbf{\widetilde{x}})$.

(2) Continuity. For any $\mathbf{x},\mathbf{\widetilde{x}}\in[0,1]^W$, let $r(\mathbf{x},\mathbf{\widetilde{x}})=\max_{w\in W}\mid x(w)-\widetilde{x}(w)\mid$ be the maximum metric. Consider the procedures (\ref{proce}) that compute $\widehat{Ch}_f(\mathbf{x})$ and $\widehat{Ch}_f(\mathbf{\widetilde{x}})$, where $r(\mathbf{x},\mathbf{\widetilde{x}})<v$. Let $(t_k)_{k=1}^L$  and $(\mathbf{z}^k)_{k=1}^L$ be the parameters in computing $\widehat{Ch}_f(\mathbf{x})$. Let $(\widetilde{t}_k)_{k=1}^{L}$,  and $(\mathbf{\widetilde{z}}^k)_{k=1}^{L}$ be the parameters in computing $\widehat{Ch}_f(\mathbf{\widetilde{x}})$.

Since $r(\mathbf{x},\mathbf{\widetilde{x}})<v$, we have $\mid t_1-\widetilde{t}_1\mid< v$ and $r(t_1\mathbf{u}^1,\widetilde{t}_1\mathbf{u}^1)< v$. Then,  $r(\mathbf{z}^1,\mathbf{\widetilde{z}}^1)=\max_{i\in\{1,2,\ldots,n\}}\mid x_i-t_1u_i^1-\widetilde{x}_i+\widetilde{t}_1u_i^1\mid<2v$.

We then establish the following inductive arguments. Suppose we have $\mid t_j-\widetilde{t}_j\mid< 2^{j-1}v$ and $r(\mathbf{z}^j,\mathbf{\widetilde{z}}^j)< 2^jv$ for all $j\in\{1,2,\ldots,k-1\}$. Then, we have $\mid\sum^{k-1}_{j=1}t_j-\sum^{k-1}_{j=1}\widetilde{t}_j\mid<(2^{k-1}-1)v$. We then consider four cases:

(i) $\sum^{k}_{j=1}t_j<1$ and $\sum^{k}_{j=1}\widetilde{t}_j<1$. Since $r(\mathbf{z}^{k-1},\mathbf{\widetilde{z}}^{k-1})< 2^{k-1}v$, we have $\mid t_k-\widetilde{t}_k\mid< 2^{k-1}v$.

(ii) $\sum^{k}_{j=1}t_j=\sum^{k}_{j=1}\widetilde{t}_j=1$. We have $\mid t_k-\widetilde{t}_k\mid< (2^{k-1}-1)v$;

(iii) $\sum^{k}_{j=1}t_j=1$ and $\sum^{k}_{j=1}\widetilde{t}_j<1$. We have $\sum^{k}_{j=1}\widetilde{t}_j=\sum^{k-1}_{j=1}\widetilde{t}_j+\widetilde{t}_k<1=\sum^{k-1}_{j=1}t_j+t_k$, and thus $t_k-\widetilde{t}_k>\sum^{k-1}_{j=1}\widetilde{t}_j-\sum^{k-1}_{j=1}t_j>-(2^{k-1}-1)v$. Since $r(\mathbf{z}^{k-1},\mathbf{\widetilde{z}}^{k-1})< 2^{k-1}v$ and $\sum^{k}_{j=1}\widetilde{t}_j<1$ also imply $t_k<\widetilde{t}_k+2^{k-1}v$, we have $\mid t_k-\widetilde{t}_k\mid< 2^{k-1}v$.

(iv) $\sum^{k}_{j=1}t_j<1$ and $\sum^{k}_{j=1}\widetilde{t}_j=1$, this is symmetric to (iii) and we also have $\mid t_k-\widetilde{t}_k\mid< 2^{k-1}v$.

According to the above inductive arguments, for each $j\in\{1,2,\ldots,L\}$, we have $\mid t_j-\widetilde{t}_j\mid< 2^{j-1}v$ and then $r(t_j\mathbf{u}^j,\widetilde{t}_j\mathbf{u}^j)< 2^{j-1}v$. We have $r(\widehat{Ch}_f(\mathbf{x}),\widehat{Ch}_f(\mathbf{\widetilde{x}}))<(2^L-1)v$.

Therefore, for any $\epsilon>0$, there exists $\delta=\epsilon/[(2^L-1)\sqrt{n}]$ such that $\mid \mathbf{x}-\mathbf{\widetilde{x}}\mid<\delta$ implies $r(\mathbf{x},\mathbf{\widetilde{x}})<\delta$, and then $r(\widehat{Ch}_f(\mathbf{x}),\widehat{Ch}_f(\mathbf{\widetilde{x}}))<\epsilon/\sqrt{n}$, which further implies $\mid\widehat{Ch}_f(\mathbf{x})-\widehat{Ch}_f(\mathbf{\widetilde{x}})\mid<\epsilon$.\footnote{$\mid\cdot\mid$ is the Euclidean distance.}

\bigskip

\quad\\
\textit{Proof of Lemma \ref{lma_between}} \quad

Since $M$ is an integral matching in $\widehat{\Gamma}$, there exists $f\in \widetilde{F}$ such that $M_f(w)=1$ for each $w\in W$, and thus $\mu$ is a matching in $\Gamma$.

We then prove that for any integral $\mathbf{x}\in\{0,1\}^W$, $\widehat{Ch}_f(\mathbf{x})=Ch_f(\mathbf{x})$. Let $j\in\{1,2,\ldots,L\}$ be the index such that $\mathbf{u}^j\leq \mathbf{x}$, and $\mathbf{u}^k\leq \mathbf{x}$ does not hold for all $k<j$. (If such $j$ does not exist, we have $\widehat{Ch}_f(\mathbf{x})=Ch_f(\mathbf{x})=\mathbf{0}$.) Then, we have $t_k=0$ for all $k<j$, $t_j=1$, and $t_k=0$ for all $j<k\leq L$. Therefore, we have $\widehat{Ch}_f(\mathbf{x})=Ch_f(\mathbf{x})=\mathbf{u}^j$.

Let $M=(M_f)_{f\in \widetilde{F}}$ be a stable integral matching in a continuum market $\widehat{\Gamma}$. Let $\mu$ be the matching in $\Gamma$ such that $\mu(w)=f$ and $w\in \mu(f)$ if $M_f(w)=1$ for each $w\in W$ and $f\in \widetilde{F}$.

(1) Individual rationality. For any $f\in F$, since $M_f$ is integral, we have $\widehat{Ch}_f(M_f)=Ch_f(M_f)$. The individual rationality of firms in $M$ in $\widehat{\Gamma}$ implies that for each $f\in F$, $M_f= \widehat{Ch}_f(M_f)=Ch_f(M_f)$. Because the set of workers available at $workers(M_f)$ is precisely $\mu(f)$, we have $\mu(f)=Ch_f(\mu(f))$. The individual rationality of workers in $M$ in $\widehat{\Gamma}$ also implies the individual rationality of workers in $\mu$.

(2) No blocking coalition. Suppose $f$ and a subset of workers $S\subseteq W$ block $\mu$ in $\Gamma$, where $ind(S)=M'$. Since $S\succ_f \mu(f)$ in $\Gamma$, we have $\widehat{Ch}_f(M'\vee M)=Ch_f(M'\vee M)\neq M$. Then by the revealed preference property, we have $\widehat{Ch}_f(M'\vee M)\succ_f M$ in $\widehat{\Gamma}$. Since $f\succeq_w \mu(w)$ for all $w\in S$ in $\Gamma$, then $\widehat{Ch}_f(M'\vee M)\leq M'\vee M\leq A^{\preceq f}(M)$. Thus $f$ and $\widehat{Ch}_f(M'\vee M)$ block $M$ in $\widehat{\Gamma}$. A contradiction.

\bigskip

\quad\\
\textit{Proof of Lemma \ref{lma_trans}} \quad

We first prove the following two lemmata. Given a subpopulation $\mathbf{x}\in[0,1]^W$, let $integral(\mathbf{x})=\mathbf{y}\in\{0,1\}^W$ where $y(w)=0$ if $x(w)=0$ and $y(w)=1$ if $x(w)>0$ for each $w\in W$. Given a pseudo-matching $M$, for each $f\in F$ we still define $B^{\preceq f}(M)$ by (\ref{DfM}). Let $A^{\prec f}(M)\in [0,1]^W$ be the subpopulation such that $A^{\prec f}(M)(w)=\sum_{f'\in \widetilde{F}:f\succ_w f'}M_{f'}(w)$ for each $w\in W$. $A^{\prec f}(M)$ is the subpopulation of workers assigned to some firm worse than $f$ in $M$ according to their preferences. Note that $A^{\prec f}(M)+M_{f}=A^{\preceq f}(M)$ for each $f\in F$.

\begin{lemma}\label{lma_block}
\normalfont
Let $M$ be a pseudo-matching and there is a firm $f'\in F$ such that $\widehat{Ch}_{f'}(M_{f'})=M_{f'}$. Consider the procedure (\ref{proce}) that computes $\widehat{Ch}_{f'}(M_{f'})=\sum_{j=1}^Lt_j\mathbf{u}^j$. There exists a blocking coalition that involves $f'$ if and only if there exists $k\in\{1,2,\ldots,L\}$ such that $\sum_{j=1}^kt_j<1$ and $\mathbf{u}^k\leq integral(A^{\prec f'}(M))$.
\end{lemma}

\begin{proof}
``$\Leftarrow$''. $\mathbf{u}^k\leq integral(A^{\prec f'}(M))$ implies that there exists $\epsilon>0$ such that $\epsilon\mathbf{u}^k\leq A^{\prec f'}(M)$. Then, the revealed preference property and $\sum_{j=1}^kt_j<1$ imply $\widehat{Ch}_{f'}(\epsilon\mathbf{u}^k+M_{f'})=\widehat{Ch}_{f'}[\widehat{Ch}_{f'}(\epsilon\mathbf{u}^k+M_{f'})\vee M_{f'}]\neq M_{f'}$. We also have $\widehat{Ch}_{f'}(\epsilon\mathbf{u}^k+M_{f'})\leq\epsilon\mathbf{u}^k+M_{f'}\leq A^{\prec f'}(M)+M_{f'}=A^{\preceq f'}(M)$. Hence, $f'$ and $\widehat{Ch}_{f'}(\epsilon\mathbf{u}^k+M_{f'})$ form a blocking coalition.

``$\Rightarrow$''. If $t_1=1$, then $M_{f'}=\mathbf{u}^1$ and there is no blocking coalition that involves $f'$. Thus there exists $k\in\{1,2,\ldots,L\}$ such that $\sum_{j=1}^kt_j<1$. Suppose for all such $k$, $\mathbf{u}^k\leq integral(A^{\prec f'}(M))$ does not hold. Then, consider the procedures (\ref{proce}) that compute $\widehat{Ch}_{f'}(M_{f'})$ and $\widehat{Ch}_{f'}(A^{\preceq f'}(M))=\sum_{j=1}^L\widehat{t}_j\mathbf{u}^j$. Since $A^{\preceq f'}(M)=M_{f'}+A^{\prec f'}(M)$, and $\mathbf{u}^k\leq integral(A^{\prec f'}(M))$ does not hold for all $k$ when $\sum_{j=1}^kt_j<1$, we know that $t_j=\widehat{t}_j$ for all $j\in\{1,2,\ldots,L\}$. Thus we have $\widehat{Ch}_{f'}(M_{f'})=\widehat{Ch}_{f'}(A^{\preceq f'}(M))$, and by revealed preference property there is no blocking coalition that involves $f'$.
\end{proof}

\begin{lemma}\label{lma_fblock}
\normalfont
Let $M$ be a stable pseudo-matching. Choose $f'\in F$ and let $M'$ be a pseudo-matching such that $M'_{f'}(w)=0$ if $M_{f'}(w)=0$ for each $w\in W$, and $M'_{f}=M_{f}$ for all $f\neq f'$. Then, there exists no blocking coalition that involves $f\neq f'$ in $M'$.
\end{lemma}

\begin{proof}
For any $f\neq f'$ consider the procedure (\ref{proce}) that computes $Ch_f(M_f)=Ch_f(M'_f)=\sum_{j=1}^Lt_j\mathbf{u}^j$. Since $M$ is stable, by Lemma \ref{lma_block} for all $k\in\{1,2,\ldots,L\}$ such that $\sum_{j=1}^kt_j<1$, $\mathbf{u}^k\leq integral(A^{\prec f}(M))$ does not hold. Since $M'_{f'}(w)=0$ if $M_{f'}(w)=0$, we have $integral(A^{\prec f}(M'))\leq integral(A^{\prec f}(M))$.
Therefore, for all $k\in\{1,2,\ldots,L\}$ such that $\sum_{j=1}^kt_j<1$, $\mathbf{u}^k\leq integral(A^{\prec f}(M'))$ does not hold for all $f\neq f'$. By Lemma \ref{lma_block}, there exists no blocking coalition that involves $f\neq f'$ in $M'$.
\end{proof}

In each of the three transformations, since we set $M'_f(w)=1$ only when $M_f(w)>0$ for each $f\in \widetilde{F}$ and each $w\in W$, the workers' individual rationality holds for $M'$ from the individual rationality of $M$. It is also straightforward to see that the firms' individual rationality also holds. By Lemma \ref{lma_fblock}, in each of the type-1 and type-2 transformations, there exists no blocking coalition that involves $f\neq f'$ in $M'$. Let $\mathbf{u}^1\succ\mathbf{u}^2\succ\cdots\succ\mathbf{u}^L\succ\mathbf{0}$ be the preference order of firm $f'$ over its acceptable sets of workers in $\Gamma$.

(1) Type-1. Suppose in $M'$, there exists a blocking coalition that involves $f'$. Consider the procedure (\ref{proce}) that computes $\widehat{Ch}_{f'}(M'_{f'})=\sum_{j=1}^Lt'_j\mathbf{u}^j$. Since $M'_{f'}=\mathbf{0}$, then $t'_j=0$ for all $j\in\{1,2,\ldots,L\}$. By Lemma \ref{lma_block}, there exists $k\in\{1,2,\ldots,L\}$ such that $\mathbf{u}^k\leq integral(A^{\prec f'}(M'))=integral(A^{\prec f'}(M))$. Consider the procedure (\ref{proce}) that computes $\widehat{Ch}_{f'}(M_{f'})=\sum_{j=1}^Lt_j\mathbf{u}^j$. Since $\sum_{j=1}^Lt_j<1$ and there exists $k\in\{1,2,\ldots,L\}$ such that $\mathbf{u}^k\leq integral(A^{\prec f'}(M))$, by Lemma \ref{lma_block}, $M$ is not stable. A contradiction.

(2) Type-2. Suppose in $M'$, there exists a blocking coalition that involves $f'$. Consider the procedure (\ref{proce}) that computes $\widehat{Ch}_{f'}(M'_{f'})=\sum_{j=1}^Lt'_j\mathbf{u}^j$. Since $M'_{f'}=\mathbf{u}^k$, by Lemma \ref{lma_block}, there exists $k'<k$ such that $\mathbf{u}^{k'}\leq integral(A^{\prec f'}(M'))$. Then,  Consider the procedure (\ref{proce}) that computes $\widehat{Ch}_{f'}(M_{f'})=\sum_{j=1}^Lt_j\mathbf{u}^j$. Since $t_k>0$, we have $\sum_{j=1}^{k'}t_j<1$, and $\mathbf{u}^{k'}\leq integral(A^{\prec f'}(M'))=integral(A^{\prec f'}(M))$. By Lemma \ref{lma_block}, $M$ is not stable. A contradiction.

(3) Type-3. If we set $M'_{{\o}}(w')=0$ for some $w'\in W$ such that $M_{{\o}}(w')\in(0,1)$, then for each $f\in F$, $B^{\preceq f}(M')\leq A^{\preceq f}(M)$. Any $f\in F$ that can form a blocking coalition with some $M''_f$ in $M'$, can also form a blocking coalition with $M''_f$ in $M$.

If we set $M'_{{\o}}(w')=1$ for some $w'\in W$ such that $M_{{\o}}(w')\in(0,1)$, then for any $f\in F$, $integral(A^{\prec f}(M))=integral(A^{\prec f}(M'))$. If there exists a blocking coalition that involves $f$ in $M'$, then by Lemma \ref{lma_block}, there also exists a blocking coalition that involves $f$ in $M$.

\subsection{Proof of Theorem \ref{thm_main}}

The proof of Theorem \ref{thm_main} is illustrated by Example \ref{illustrate}. Every matching market $\Gamma$ induces a continuum market $\widehat{\Gamma}$, where Lemma \ref{lma_cont} and the existence theorem of CKK indicate that a stable matching $M$ is guaranteed to exist in $\widehat{\Gamma}$. For each $f\in F$, consider the procedure (\ref{proce}) that computes $\widehat{Ch}_{f}(M_{f})=\sum_{j=1}^Lt_j\mathbf{u}^{j}=M_{f}$, where the last equality is implied by the individual rationality of firms in $M$. We then construct a matrix $B^{*f}$ as follows.

\begin{equation}\label{Af}
B^{*f}=\left\{
\begin{aligned}
&[\mathbf{u}^{k_1},\mathbf{u}^{k_2},\ldots,\mathbf{u}^{k_s}], &\quad\text{if} \quad \sum_{j=1}^Lt_j=1.\\
&[\mathbf{u}^{k_1},\mathbf{u}^{k_2},\ldots,\mathbf{u}^{k_s},\mathbf{0}], &\quad\text{if} \quad \sum_{j=1}^Lt_j<1.
\end{aligned}
\right.
\end{equation}
where $t_l>0$ for each $l\in\{k_1,k_2,\ldots,k_s\}$ and $k_1<k_2<\cdots<k_s$.\footnote{In (\ref{Af}), for each $l\in\{k_1,k_2,\ldots,k_s\}$, $\mathbf{u}^l$ is a column of matrix $B^{*f}$. $B^{*f}$ is an $n\times s$ or $n\times (s+1)$ matrix. The numbers $\{k_1,k_2,\ldots,k_s\}$ and $s$ are different for different firms.}

We then construct matrix $B^{{\o}}$ and $B^f$ for each $f\in F$. Each column of $B^{{\o}}$ is an $(m+n)$- dimensional unit vector, where the $(m+j)$-th unit vector is in $B^{{\o}}$ if $M_{{\o}}(w_j)>0$.\footnote{The $(m+j)$-th unit vector is the $(m+n)$-dimensional unit vector with its $(m+j)$-th coordinate being 1. In Example \ref{illustrate}, $B^{{\o}}$ is the matrix with a single column $B_5$.} We construct $B^f$ from $B^{*f}$ for each $f\in F$. Each $B^f$ has $m+n$ rows, where the first $m$ components of each column are the components of the $i$-th unit vector of dimension $m$ when $f=f_i$. The last $n$ components of each column of $B^f$ are those of each column of $B^{*f}$, thus the number of columns of each $B^f$ is $s$ or $s+1$ in its corresponding formula (\ref{Af}). Then, let matrix $B=[B^{f_1},B^{f_2},\cdots,B^{f_m},B^{{\o}}]$.
Let $cl(f)$ be the number of columns of $B^f$ for each $f\in \widetilde{F}$. Thus, $B$ is an $(m+n)\times\sum_{f\in \widetilde{F}}cl(f)$ matrix. Theorem \ref{thm_main} is then implied by the following lemma and Lemma \ref{lma_between}.

\begin{lemma}
\normalfont
Consider the system of linear equations $B\mathbf{z}=\mathbf{1}$, where $\mathbf{z}$ is a $\sum_{f\in \widetilde{F}}cl(f)$-dimensional vector, and $\mathbf{1}$ is the $(m+n)$-dimensional vector with all its coordinates being 1. If the firms' demand type $\mathcal{D}$ is totally unimodular, then there exists a solution $\mathbf{z}=(\mathbf{z}^{f_1},\mathbf{z}^{f_2},\ldots,\mathbf{z}^{f_m},\mathbf{z}^{{\o}})\in\{0,1\}^{\sum_{f\in \widetilde{F}}cl(f)}$ to $B\mathbf{z}=\mathbf{1}$, where $\mathbf{z}^f$ is a $cl(f)$-dimensional vector for each $f\in \widetilde{F}$.\footnote{By $\mathbf{z}=(\mathbf{z}^{f_1},\mathbf{z}^{f_2},\ldots,\mathbf{z}^{f_m},\mathbf{z}^{{\o}})$ we mean the first $cl(f_1)$ components of $\mathbf{z}$ are the components of $\mathbf{z}^{f_1}$, the $(cl(f_1)+1)$-th to $(cl(f_1)+cl(f_2))$-th components are the components of $\mathbf{z}^{f_2}$, and so on. In Example \ref{illustrate}, $\mathbf{z}=(1,0,0,1,1)$ is such a solution where $\mathbf{z}^{f_1}=(1,0)$, $\mathbf{z}^{f_2}=(0,1)$, and $\mathbf{z}^{{\o}}=1.$} Moreover, $M'$ is a stable integral matching in $\widehat{\Gamma}$, where $M'_{f}=\sum_{j=1}^{cl(f)}z^f_jB_j^{*f}$ for each $f\in \widetilde{F}$.\footnote{$B_j^{*f}$ is the vector of the $j$-th column of $B^{*f}$.}
\end{lemma}

\begin{proof}
We first show that the polytope $\{\mathbf{z}\mid B\mathbf{z}=\mathbf{1}, \mathbf{z}\geq0\}$ is nonempty. For each $f\in F$, consider the procedure (\ref{proce}) that computes $\widehat{Ch}_{f}(M_{f})=\sum_{j=1}^Lt_j\mathbf{u}^{j}=M_{f}$ and let
\begin{equation*}
\mathbf{\widehat{z}}^f=\left\{
\begin{aligned}
&(t_{k_1},t_{k_2},\ldots,t_{k_s}), &\quad\text{if} \quad \sum_{j=1}^Lt_j=1.\\
&(t_{k_1},t_{k_2},\ldots,t_{k_s},1-\sum_{i=1}^st_{k_i}), &\quad\text{if} \quad \sum_{j=1}^Lt_j<1.
\end{aligned}
\right.
\end{equation*}
where $k_1,k_2,\ldots,k_s$ are those in (\ref{Af}). Let $\mathbf{\widehat{z}}^{{\o}}$ be the $cl({\o})$-dimensional vector where $\widehat{z}_l^{{\o}}=M_{{\o}}(w_j)$ if $B_l^{{\o}}$ is the $(m+j)$-th unit vector for each $l\in\{1,\ldots,cl({\o})\}$.
Let $B^{**}$ be the matrix constituted of the first $m$ rows of $B$, and $B^*$ the matrix constituted of the last $n$ rows of $B$. Let $\mathbf{\widehat{z}}=(\mathbf{\widehat{z}}^{f_1},\mathbf{\widehat{z}}^{f_2},\ldots,\mathbf{\widehat{z}}^{f_m},\mathbf{\widehat{z}}^{{\o}})$ . $\sum_{i=1}^{cl(f)}\widehat{z}^f_i=1$ for each $f\in F$ implies $B^{**}\mathbf{\widehat{z}}=\mathbf{1}$, and $\sum_{f\in{\widetilde{F}}}M_f(w)=1$ for each $w\in W$ implies $B^*\mathbf{\widehat{z}}=\mathbf{1}$. Hence, we know $\mathbf{\widehat{z}}$ is in $\{\mathbf{z}\mid B\mathbf{z}=\mathbf{1}, \mathbf{z}\geq0\}$.

If matrix $B$ is unimodular, all vertices of the polytope $\{\mathbf{z}\mid B\mathbf{z}=\mathbf{1}, \mathbf{z}\geq0\}$ are integral (\citealp{HK56}; see also Theorem 21.5 of \citealp{S86}). Then, since the polytope is nonempty, we know that there is at least an integral vertex on this polytope. We now show that $B$ is unimodular if the firms' demand type $\mathcal{D}$ is totally unimodular. For any linearly independent subset $\widehat{B}$ of columns from $B$, we partition $\widehat{B}$ into $\widehat{B}=\cup_{f\in \widetilde{F}}\widehat{B}^f$, where $\widehat{B}^f$ is the collection of vectors of $\widehat{B}$ from $B^{f}$ for each $f\in \widetilde{F}$. $\widehat{B}^f$ is possibly empty for some $f\in \widetilde{F}$.\footnote{For instance, if we consider $\widehat{B}=\{B_1,B_2,B_5\}$ in Example \ref{illustrate}, then $\widehat{B}^{f_2}$ is empty. In the following extension to a basis for $\mathbb{R}^{m+n}$ from $\widehat{B}$, we put the $i$-th unit vector of $m+n$ dimensions into $\widehat{B}$ when $\widehat{B}^{f_i}$ is empty for each $i\in\{1,\ldots,m\}$.} Let $\mathbf{b}^f_1,\mathbf{b}^f_2,\cdots$ denote the elements of $\widehat{B}^f$ for each $f\in \widetilde{F}$ if $\widehat{B}^f$ is not empty. For each $i\in\{1,\ldots,m\}$, if $\widehat{B}^{f_i}$ is empty, let $\mathbf{b}^{f_i}_1$ be the $i$-th unit vector of $m+n$ dimensions. For each $f\in F$, let $B^f=\{\mathbf{b}^f_2-\mathbf{b}^f_1,\mathbf{b}^f_3-\mathbf{b}^f_1,\cdots\}$ if there are at least 2 elements in $\widehat{B}^f$, and $B^f=\emptyset$ otherwise. Let $\widetilde{B}_F=(\cup_{f\in F}B^f)$ and $\widetilde{B}=(\cup_{f\in F}B^f)\cup \widehat{B}^{\o}$. Note that the first $m$ coordinates of each vector from $\widetilde{B}_F$ and $\widetilde{B}$ are 0. Let $\widetilde{B}^*_F$ and $\widetilde{B}^*$ be the set of $n$-dimensional vectors by removing the first $m$ components of each vector from $\widetilde{B}_F$ and $\widetilde{B}$, respectively. If $\mathcal{D}$ is totally unimodular, then $\widetilde{B}^*_F$ is totally unimodular. According to the reason stated in footnote \ref{footnote}, $\widetilde{B}^*$ is unimodular. Thus, $\widetilde{B}^*$ can be extended to a basis for $\mathbb{R}^n$,  of integer vectors, with determinant $\pm1$. Then, the set $\{\mathbf{b}^{f_1}_1,\mathbf{b}^{f_2}_1,\cdots,\mathbf{b}^{f_m}_1\}\cup\widetilde{B}$ can be extended to a basis for $\mathbb{R}^{m+n}$,  of integer vectors, with determinant $\pm1$. Since adding one column to another column leaves the determinant unchanged, such extension also exists for $\{\mathbf{b}^{f_1}_1,\mathbf{b}^{f_2}_1,\cdots,\mathbf{b}^{f_m}_1\}\cup\widehat{B}$, which can be extended from $\widehat{B}$. Therefore, $B$ is unimodular.

Now we know there is at least an integral vertex on the polytope $\{\mathbf{z}\mid B\mathbf{z}=\mathbf{1}, \mathbf{z}\geq0\}$ when the firms' demand type $\mathcal{D}$ is totally unimodular. According to the structure of matrix $B$, any nonnegative integral solution to $B\mathbf{z}=\mathbf{1}$ must be a 0-1 vector. Let $\mathbf{z}=(\mathbf{z}^{f_1},\mathbf{z}^{f_2},\ldots,\mathbf{z}^{f_m},\mathbf{z}^{{\o}})$ be an integral point of $\{\mathbf{z}\mid B\mathbf{z}=\mathbf{1}, \mathbf{z}\geq0\}$, where $\mathbf{z}^f$ is a $cl(f)$-dimensional vector for each $f\in \widetilde{F}$. Let $M'$ be the integral pseudo-matching in $\widehat{\Phi}$, where $M'_{f}=\sum_{j=1}^{cl(f)}z^f_jB_j^{*f}$ for each $f\in \widetilde{F}$. Because $B^{**}\mathbf{z}=\mathbf{1}$ (i.e., $\sum_{j=1}^{cl(f)}z^f_j=1$ for each $f\in F$), we know that $M'$ can be obtained via stable transformations on $M$. By Lemma \ref{lma_trans}, $M'$ is a stable integral pseudo-matching. Since $B^{*}\mathbf{z}=\mathbf{1}$ implies $\sum_{f\in{\widetilde{F}}}M'_f(w)=1$ for each $w\in W$, we know that $M'$ is a stable integral matching.
\end{proof}

\subsection{Proof of Theorem \ref{thm_special}}\label{proof_special}

We present the proof first and then an example to illustrate the proof.

Suppose the firms have unit-demand preferences over a technology tree $T=(V,E,W)$. Let $G=(V,E')$ be a complete directed graph where the set of vertices $V$ is the same as that of $T$, and the direction of each edge $e\in E'$ is arbitrary. We use $G$ and the directed tree $(V,E)$ to define matrix $H$ as follows. For each $e\in E$ and $e'=(v,v')\in E'$,
\begin{align*}
H_{e,e'}=+&1\quad\text{if the unique $v-v'$ path in $(V,E)$ passes through $e$ forwardly;}\\
-&1\quad\text{if the unique $v-v'$ path in $(V,E)$ passes through $e$ backwardly;}\\
&0\quad\text{if the unique $v-v'$ path in $(V,E)$ does not pass through $e$.}
\end{align*}
$H$ is a network matrix, which is totally unimodular (see e.g., Chapter 19.3 of \citealp{S86}).

Because each worker is a specialist in $T$, we can define $e^w$ to be the unique edge that $w$ engages for each $w\in W$. We use $G$ and the technology tree $T=(V,E,W)$ to define matrx $H'$ as follows. For each $w\in W$ and $e'=(v,v')\in E'$,
\begin{align*}
H'_{w,e'}=+&1\quad\text{if the unique $v-v'$ path in $(V,E)$ passes through $e^w$ forwardly;}\\
-&1\quad\text{if the unique $v-v'$ path in $(V,E)$ passes through $e^w$ backwardly;}\\
&0\quad\text{if the unique $v-v'$ path in $(V,E)$ does not pass through $e^w$.}
\end{align*}

$H'$ can be obtained from $H$ by repeating some rows. In particular, if $\mid W^e\mid=k$ for edge $e$ in the technology tree, then the row of edge $e$ of $H$ is repeated $k$ times in $H'$. Thus we know that $H'$ is also totally unimodular.

If $\mathbf{d}$ belongs to the demand type of a firm $f\in F$, then either $\mathbf{d}$ is a column of $H'$, or $-\mathbf{d}$ is a column of $H'$ (see an illustration below). Therefore, $\mathcal{D}$ is totally unimodular.

For instance, consider the technology tree in Example \ref{exam_app}, and let $G=(V,E')$ be the following complete directed graph.

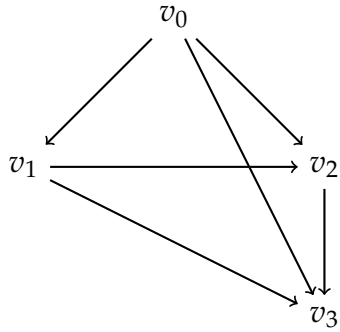
\begin{center}
\begin{tikzpicture}[thick,->]
	\node (v0) at (0,0) {$v_0$};
	\node (v1) at (-2,-2)  {$v_1$};
    \node (v2) at (2,-2)  {$v_2$};
	\node (v3) at (2,-4)  {$v_3$};
\draw (v0)--(v1);
\draw (v0)--(v2);
\draw (v0)--(v3);
\draw (v1)--(v2);
\draw (v1)--(v3);
\draw (v2)--(v3);		
	\end{tikzpicture}
\end{center}
The directed graph $G$ and the directed tree $(V,E)$ generate matrix $H$ as follows.
\begin{center}
\begin{tabular}
[c]{c|cccccc}
& $v_0v_1$ & $v_0v_2$ & $v_0v_3$& $v_1v_2$& $v_1v_3$& $v_2v_3$\\\hline
$v_0v_1$ & 1 & 0 & 0& -1& -1& 0\\
$v_0v_2$ & 0 & 1 & 1& 1& 1& 0\\
$v_2v_3$ & 0 & 0 & 1& 0& 1& 1
\end{tabular}
\end{center}
The directed graph $G$ and the technology tree $T=(V,E,W)$ generate matrix $H'$ as follows.
\begin{center}
\begin{tabular}
[c]{c|cccccc}
& $v_0v_1$ & $v_0v_2$ & $v_0v_3$& $v_1v_2$& $v_1v_3$& $v_2v_3$\\\hline
$w_1$ & 1 & 0 & 0& $-1$& $-1$& 0\\
$w_2$ & 1 & 0 & 0& $-1$& $-1$& 0\\
$w_3$ & 0 & 1 & 1& 1& 1& 0\\
$w_4$ & 0 & 0 & 1& 0& 1& 1
\end{tabular}
\end{center}
The matrix $H'$ is exactly the one by repeating the first row of the matrix $H$. Now we consider the firms' preferences of (\ref{pre_app}). Let $H'_j$ be the $j$-th column of matrix $H'$. We find that $f_1$'s demand type is $\{H'_1,H'_2,-H'_4\}$, and $f_2$'s demand type is $\{H'_1,H'_3,H'_5\}$.

\end{document}